\documentclass[journal,10pt]{IEEEtran}

\usepackage[utf8]{inputenc} 
\usepackage[T1]{fontenc}
\usepackage{url}
\usepackage{ifthen}
\usepackage{mdframed}
\usepackage{cite}
\usepackage{enumerate}
\usepackage[cmex10]{amsmath}
\usepackage{graphicx}
\usepackage{subfig}
\usepackage{dsfont}
\usepackage{mathtools}
\usepackage{enumitem}

\usepackage{amssymb, amsmath, amsthm, amsfonts}
\usepackage{xcolor}
\newtheorem{thm}{Theorem}[section]
\newtheorem{lem}[thm]{Lemma}
\newtheorem{prop}[thm]{Proposition}
\newtheorem{cor}[thm]{Corollary}

\newtheorem{assumption}{Assumption}
\newenvironment{Assumption}
{\begin{mdframed}\begin{assumption}}
		{\end{assumption}\end{mdframed}}

\def\u{\mathbf{u}}
\def\v{\mathbf{v}}
\def\0{\mathbf{0}}

\def\bZ{\mathbf{Z}}

\def\bz{\mathbf{z}}
\def\1{\mathbf{1}}

\def\E{\mathbb{E}}
\def\P{\mathbb{P}}
\def\R{\mathbb{R}}
\def\Z{\mathbb{Z}}

\def\cR{\mathcal{R}}

\def\cG{\mathcal{G}}

\def\G{\Gamma}
\def\Palm{\P^\0}
\def\Palmexp{\E^\0}
\def\MarkPalm{\E^{(\0,\1)}}
\def\ProbMark{\P^{(\0,\1)}}
\def\Ball{B_1(\0)}
\theoremstyle{definition}
\newtheorem{defn}{Definition}[section]

\newcommand{\remove}[1]{}

\begin{document}
	
	\title{An Analysis of Probabilistic Forwarding \\ of Coded Packets on Random Geometric Graphs}

	\author{%
		B.~R.~Vinay Kumar,~\IEEEmembership{Student Member,~IEEE},
		Navin Kashyap,~\IEEEmembership{Senior Member,~IEEE},
		and D.~Yogeshwaran%
		\thanks{This work was presented in part at the 19th International Symposium on Modeling and Optimization in Mobile, Ad hoc, and Wireless Networks (WiOpt 2021), which was held virtually during Oct 18--21, 2021.}
		\thanks{B.~.R.~Vinay Kumar and N.~Kashyap are with the Department of Electrical Communication Engineering, 
			Indian Institute of Science, Bengaluru, India.
			Email: \{vinaykb, nkashyap\}@iisc.ac.in}
		\thanks{D.~Yogeshwaran is with the Indian Statistical Institute, Bengaluru, India. Email: {d.yogesh@isibang.ac.in}}
		\thanks{The work of B.R.\ Vinay Kumar was supported in part by a fellowship from the Centre for Networked Intelligence (a Cisco CSR initiative) of the Indian Institute of Science. The work of N.\ Kashyap was supported in part by a Swarnajayanti Fellowship awarded by the Dept.\ of Science \& Technology, Govt.\ of India. The work of D.\ Yogeshwaran was supported in part by SERB-MATRICS grant MTR/2020/000470.}
		}

		\maketitle
		
		\begin{abstract}
			We consider the problem of energy-efficient broadcasting on large ad-hoc networks. Ad-hoc networks are generally modeled using random geometric graphs (RGGs). Here, nodes are deployed uniformly in a square area around the origin, and any two nodes which are within Euclidean distance of $1$ are assumed to be able to receive each other's broadcast. A source node at the origin encodes $k$ data packets of information into $n\  (>k)$ coded packets and transmits them to all its one-hop neighbors. The encoding is such that, any node that receives at least $k$ out of the $n$ coded packets can retrieve the original $k$ data packets. Every other node in the network follows a probabilistic forwarding protocol; upon reception of a previously unreceived packet, the node forwards it with probability $p$ and does nothing with probability $1-p$. We are interested in the minimum forwarding probability which ensures that a large fraction of nodes can decode the information from the source. We deem this a \emph{near-broadcast}. The performance metric of interest is the expected total number of transmissions at this minimum forwarding probability, where the expectation is over both the forwarding protocol as well as the realization of the RGG. In comparison to probabilistic forwarding with no coding, our treatment of the problem indicates that, with a judicious choice of $n$, it is possible to reduce the expected total number of transmissions while ensuring a near-broadcast.
		\end{abstract}
		\renewcommand{\thefootnote}{\arabic{footnote}}
		\setcounter{footnote}{0}
		\section{Introduction}
		Ad-hoc networks are distributed networks with no centralized infrastructure. Applications involving the Internet of Things (IoT), such as healthcare, smart factories and homes, intelligent transport etc., have lead to wide-spread presence of dense ad-hoc networks. Individual nodes in these networks are typically low-cost and energy-constrained, having limited computational ability and knowledge of the network topology.
		
		Random network models have found wide acceptance in modeling wireless ad-hoc networks. In particular, random geometric graphs (RGGs) have been used in the literature to model spatially distributed networks (see e.g. \cite{vaze2015random} and \cite{franceschetti2008random}). These are generated by scattering (a Poisson number of) nodes in a finite area uniformly at random and connecting nodes within a pre-specified distance. The random distribution of nodes captures the variability in the deployment of the nodes of an ad-hoc network. The distance threshold conforms to the maximum range at which a transmission from a node, with maximum power, is received reliably. A more formal description of our network setting is provided in the next section.
		
		Exchange of network-critical information for network control and routing happens primarily through broadcast mechanisms in these networks. A considerable number of broadcast mechanisms have been proposed in the literature (see e.g. \cite{williams2002comparison}, \cite{fragouli2008efficient} and \cite{wang2013coding}, and the references therein). Algorithms such as flooding, although being light-weight and easy to implement, give rise to unnecessary transmissions and hence are not energy efficient. Flooding is also known to result in the `broadcast-storm' problem (see \cite{tseng2002broadcast}).
		
		Probabilistic forwarding as a broadcast mechanism (see e.g., \cite{forero2019thesis}, \cite{sasson2003probabilistic}, \cite{haas2006gossip}) has been proposed in the literature as an alternative to flooding. Here, each node, on receiving a packet for the first time, either forwards it to all its one-hop neighbours with probability $p$ or takes no action with probability $1-p$. While this mechanism reduces the number of transmissions, reception of a packet by a network node is not guaranteed.
		
		To improve the chances of a network node receiving a packet and to handle packet drops, we introduce coding along with probabilistic forwarding. Let us suppose that the source possesses $k_s$ message packets which need to be broadcast. These $k_s$ message packets are first encoded into $n$ coded packets such that, for some $k \ge k_s$, the reception of any $k$ out of the $n$ coded packets by a node, suffices to retrieve the original $k_s$ message packets. Examples of codes with this property are Maximum Distance Separable (MDS) codes ($k=k_s$), fountain codes ($k=k_s(1+\epsilon)$ for some $\epsilon >0$) etc. which are used in practice.
		
		The $n$ coded packets are indexed using integers from $1$ to $n$, and the source transmits each packet to all its one-hop neighbours. Every other node in the network, upon reception of a packet (say packet $\# j$) uses the probabilistic forwarding mechanism described above. The node ignores all subsequent receptions of packet $\# j$. Packet collisions and interference effects are neglected.
		
		In this paper, we analyze the performance of the above algorithm on RGGs. In particular, we wish to find the minimum retransmission probability $p$ for which the expected fraction of nodes receiving at least $k$ out of the $n$ coded packets is close to 1, which we deem a ``near-broadcast''. Here, it is to be clarified that the expectation is over both the realization of the RGG and the probabilistic forwarding protocol. This probability yields the minimum value for the expected total number of transmissions across all the network nodes needed for a near-broadcast. The expected total number of transmissions is taken to be a measure of the energy expenditure in the network. 
		
		In our previous work \cite{ToN}, we have analyzed the probabilistic forwarding mechanism described here on deterministic graphs such as trees and grids. It was found that, introducing coded packets with probabilistic forwarding, offered significant energy benefits in terms of the number of transmissions needed for a near-broadcast on well-connected graphs such as grids and other lattice structures. However, for $d$-regular trees, such energy savings were not observed. RGGs (in the super-critical regime) show similar behaviour as grids, i.e., for an intelligently chosen value of the number of coded packets, $n$, and the minimum forwarding probability, the energy expenditure in the network is considerably lesser for a near-broadcast, when compared to the scenario of probabilistic forwarding with no coding.
		
		In this paper, we justify these observations using rigorous methods. While the techniques used here are similar to the ones on the grid (in \cite{ToN}), the additional complications due to the randomness of the underlying graph need to be addressed. This calls for the use of  ideas from continuum percolation, ergodic theory and Palm theory to circumvent some of the technicalities encountered. These mathematical techniques could be of independent interest for related problems. Moreover, our method of analysis may also extend to more general broadcasting models and other point processes.
		
		
		The rest of the paper is organized as follows. In Section \ref{sec:formulation}, we describe our network setup and formulate our problem. Section \ref{sec:simu} provides the simulation results of the probabilistic forwarding algorithm on RGGs. In Section \ref{sec:prelims}, we provide definitions and notations of RGGs on $\R^2$. Marked point processes (MPPs) are introduced to model probabilistic forwarding on the RGG. Section \ref{sec:MPP_fwding} relates probabilistic forwarding and marked point processes. Ergodic theorems on MPPs are used to obtain some key quantities. These will serve as the main ingredients in obtaining our estimates for the minimum forwarding probability and the expected total number of transmissions which are presented in Section \ref{sec:results}. Since the estimates for the minimum forwarding probability are not computable, in Section \ref{sec:heuristic}, we provide a heuristic approach which is used to compare with the simulation results. Section \ref{sec:disc} discusses some aspects related to the assumptions and our results. The appendix contains technical results pertaining to the Palm expectations and the proof of one of our main theorems.
		
		\section{Problem formulation}\label{sec:formulation}
		We begin by describing our setting for the specific case of random geometric graphs. This introduces additional notation specific to RGGs as well.
		\subsection{Network setup} \label{sec:setup}
		A random geometric graph is parametrized by the intensity $\lambda$ and the distance threshold $r$. It suffices to study them by keeping one of the parameters fixed. In our treatment, we will fix the distance parameter $r$ to be equal to $1$, and study various properties as a function of the intensity, $\lambda$.
		
		Construct a random geometric graph $G_m$ with intensity $\lambda$ and distance threshold $r = 1$ on $\G_m := \left[\frac{-m}{2}, \frac{m}{2}\right]^2$ as follows: 
		\begin{itemize}
			\item \textbf{Step 1:} Sample the number of points, $N$, from a Poisson distribution with mean $\lambda \nu(\G_m)$. Here, $\nu(\cdot)$ is the Lebesgue measure on $\R^2$. Therefore, $N \sim \mathrm{Poi}(\lambda m^2)$. 
			\item \textbf{Step 2:} Choose points $X_1,X_2, \cdots, X_N$ uniformly and independently from $\G_m$. These form the points of a Poisson point process (see \cite[Section 2.5]{chiu2013stochastic}) $\Phi$, and constitute the vertex set of $G_m$. 
			\item \textbf{Step 3:} Place an edge between any two vertices which are within Euclidean distance $r=1$ of each other.
		\end{itemize}
		
		To carry out probabilistic forwarding over $G_m$, we need to fix a source. For this, we will assume that there is a point at the origin $\0 =(0,0) \in \R^2$. More specifically, a graph $G_m^\0$ is created with the underlying point process $\Phi^\0 \triangleq \Phi \cup \{\0\}$, as the vertex set and introducing additional edges from $\0$ to nodes which are within $\Ball$, to the edge set of $G_m$. Here, $\Ball$ (more generally, $B_1(\v)$ for $\v \in \R^2$) is a closed Euclidean ball of radius $1$ centered at $\0$ ($\v$).
		
		The inclusion of an additional point at the origin $\0$ means that all the probabilistic computations need to be made with respect to the Palm probability given a point at the origin. We direct the reader to \cite[Ch. 1.4]{baccelli2010stochastic} for an in-depth treatment of Palm theory. Heuristically, the Palm probability must be interpreted as the probability conditional on the event that the origin is a point of the point process. We denote the Palm probability by $\P^\0$ and the expectation with respect to it by $\E^\0$.
		
		The origin here is a distinguished vertex. Broadcasts initiated from it can be received by the nodes which are present in the component of the origin only. Denote by $C_\0 \equiv C_\0(G_m^\0)$, the set of nodes in the component of the origin in $G_m^\0$. The component of the origin in $G_m^\0$ forms the underlying connected graph, which we denote by $G$.
		
		

		\subsection{Probabilistic forwarding on RGG}
		\par Equipped with the underlying network, $G$, we now describe the probabilistic forwarding algorithm on it. The source, $\0$, encodes $k_s$ message packets into $n$ coded packets and transmits it to all its one-hop neighbours. Every other node in the network follows the probabilistic forwarding protocol. A node receiving a particular packet for the first time, forwards it to all its one-hop neighbours with probability $p$ and takes no action with probability $1-p$. Each packet is forwarded independently of other packets and other nodes. The node ignores all subsequent receptions of the same packet, irrespective of the decision it took at the time of first reception.
		
		\par We are interested in the following scenario. Let $R_{k,n}(G)$ be the number of nodes in $C_\0$ that receive at least $k$ out of the $n$ coded packets in $G$. We refer to these as \emph{successful receivers}. We sometimes denote this by $R_{k,n}(G_m^\0)$ to explicitly bring out the dependence on $m$. Given a $\delta >0$, we are interested in the minimum forwarding probability $p$, such that the expected fraction of successful receivers is at least $1-\delta$. The expectation here is over the probabilistic forwarding protocol for a fixed realization of $G$. In reality, the proposed broadcasting algorithm of probabilistic forwarding with coded packets, should give a good performance for any realization of the underlying graph. In other words, we would want the expected fraction of successful receivers to be at least $1-\delta$, for every realization of $G$. However, in our formulation we relax this condition by asking for it only in an expected sense. More specifically, we define
		\begin{equation}
		p_{k,n,\delta} = \inf \left\{p \ \ \Big{|} \ \ \mathbb{E} \left[ \frac{R_{k,n}(G_m^\0)}{|C_\0(G_m^\0)|} \right] \geq 1-\delta \right\},
		\label{Eq:pkndelta}
		\end{equation}
		where the expectation is over both the graph $G_m^\0$ as well as the probabilistic forwarding mechanism. Note that, from our construction, $R_{k,n}(G) = R_{k,n}(G_m^\0) \subseteq C_\0(G_m^\0)$. The number of successful receivers is normalized by the total number of vertices in $G$, which is the same as the number of vertices within the component of the origin, $|C_\0(G_m^\0)|$.
		
		The performance measure of interest, denoted by $\tau_{k,n,\delta}$, is the expected total number of transmissions across all nodes when the forwarding probability is set to $p_{k,n,\delta}$. Here, it should be clarified that whenever a node forwards (broadcasts) a packet to all its one-hop neighbours, it is counted as a single (simulcast) transmission. Our aim is to determine, for a given $k$ and $\delta$, how $\tau_{k,n,\delta}$ varies with $n$, and the value of $n$ at which it is minimized (if it is indeed minimized). To this end, it is necessary to first understand the behaviour of $p_{k,n,\delta}$ as a function of $n$. In subsequent sections, we will formulate the probabilistic forwarding mechanism as a marked point process and use results from ergodic theory to obtain the expected value of the number of successful receivers and the overall number of transmissions.
		
		\section{Simulation results}\label{sec:simu}
		Simulations were performed on an RGG generated with $m =101$ and intensity $\lambda =4.5$ and $4$. As stated before, the distance threshold parameter $r$ was set to $1$. The probabilistic forwarding mechanism was carried out with $k=20$ packets and $n$ varying from $20$ to $40$. The value of $\delta$ was set to $0.1$. Twenty realizations of $G$ were generated and $10$ iterations of the probabilistic forwarding mechanism was carried out on each of the realizations. The fraction of successful receivers was averaged over each iteration and realization of the graph.  This was used to find the minimum forwarding probability, $p_{k,n,\delta}$, required for a near-broadcast, which is plotted in Figure \ref{fig:RGG_simu}(a). The $p_{k,n,\delta}$ values so obtained were further used to find the expected total number of transmissions over the same realizations. The expected total number of transmissions $\tau_{k,n,\delta}$, normalized by $\lambda m^2$, which is the average number of points within $\G_m$, is shown in Figure \ref{fig:RGG_simu}(b). This can be interpreted as the average number of transmissions per node in the graph.
		
		Notice that the expected number of transmissions decreases initially to a minimum and then increases. The decrease indicates the benefit of introducing coding along with probabilistic forwarding. The number of coded packets, $n$, and the probability, $p_{k,n,\delta}$, corresponding to the minimum point of Figure \ref{fig:RGG_simu}(b) are the ideal parameters for operating the network to obtain maximum energy benefits.
		\begin{figure} 
			\centering
			\subfloat[Minimum retransmission probability]{%
				\includegraphics[width=\linewidth]{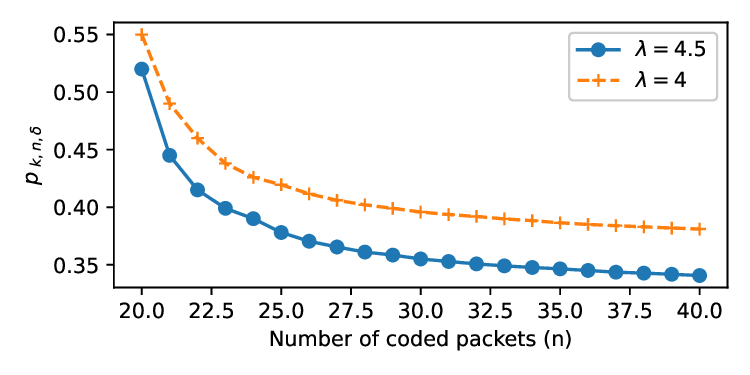}}
			\label{4a}
			\subfloat[Expected total number of transmissions]{%
				\includegraphics[width=\linewidth]{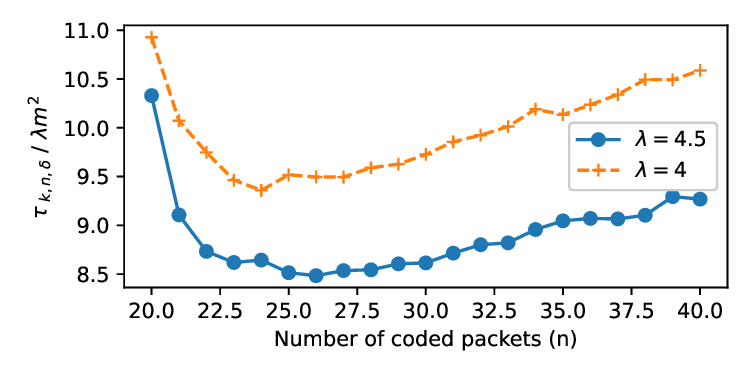}}
			\label{4b}\\
			\caption{Simulations on a random geometric graph generated on $\G_{101}$ with intensity $\lambda$ and distance threshold $r=1$. Probabilistic forwarding done with $k=20$ packets and $\delta=0.1$.}
			\label{fig:RGG_simu} 
			\vspace{-1.5em}
		\end{figure}
		
		Further, it can be observed from Fig. \ref{fig:RGG_simu}(a), that the minimum forwarding probability, $p_{k,n,\delta}$, decreases to $0$ with $n$. This is formalized in the following lemma.
		\begin{lem}
			For fixed values of $k$ and $\delta$, 
			\begin{enumerate}
				\item[\emph{(a)}] $p_{k,n,\delta}$ is a non-increasing function of n.
				\item[\emph{(b)}] $p_{k,n,\delta} \rightarrow 0$ as $n\rightarrow \infty.$
			\end{enumerate}
		\end{lem}
		The proof is on similar lines as that for deterministic graphs expounded in \cite{ToN}. Conditioning on the underlying point process, $\Phi$, gives a deterministic graph, on which the result for deterministic graphs can be used. We omit the details here.
		
		\section{Preliminaries}\label{sec:prelims} 
		In this section, we introduce the tools required to characterize the performance of the probabilistic forwarding algorithm. The probabilistic forwarding mechanism on the RGG is modeled using marked point processes which are described here. 
		
		\subsection{Random geometric graphs on $\R^2$}
		\label{sec:infrgg}
		
		\label{PPP}

		
		
		Our approach to analyzing the probabilistic forwarding mechanism on $G$ is to relate it to the probabilistic forwarding mechanism on a RGG generated on the whole $\R^2$ plane with the origin as the source. This means that the vertex set of the RGG is a Poisson point process, $\Phi$, on $\R^2$. We refer the reader to  \cite{franceschetti2008random} or \cite{baccelli2010stochastic} for the background needed on Poisson point processes. In particular, we use the procedure outlined in \cite[Section 1.3]{baccelli2010stochastic} to construct the RGG on the whole $\R^2$ plane.
		
		Create a tiling of the $\R^2$ plane with translations of $\G_m$, i.e., $\G_{i,j}:=(im,jm)\  +\  \G_m$ for $i,j \in \Z$. On each such translation, $\G_{i,j}$, construct an independent copy of a Poisson point process with intensity $\lambda$ as described in steps 1 and 2 of Section \ref{sec:setup}. 
		The random geometric graph ($\cG$) is constructed by connecting vertices which are within distance $1$ of each other. We then say $\cG \sim RGG(\lambda, 1)$.
		
		It is known that the $RGG(\lambda , 1)$ model on $\R^2$ shows a phase transition phenomenon (see e.g. \cite{penrose2003random}). For $\lambda > \lambda_c$, the \emph{critical intensity}, there exists a unique infinite cluster, $C \equiv C(\Phi)$, in the RGG almost surely. The value of $\lambda_c $ is not exactly known, but simulation studies such as \cite{quintanilla2000efficient} indicate that $\lambda_c \approx 1.44$. The \emph{percolation probability} $\theta(\lambda)$ is defined as the probability that the origin is present in the infinite cluster $C$, i.e., $\theta(\lambda) := \Palm(\0 \in C)$. We remark here that there is no known analytical expression for $\theta(\lambda)$ nor are there good approximations. Since we are interested in large networks, we will assume throughout our analysis that we operate in the super-critical region, i.e., $\lambda>\lambda_c$.
		
		\subsection{Marked Point Process}
		During the course of the probabilistic forwarding protocol on the RGG, each node decides independently whether to forward a particular packet with probability $p$. Marked point processes (MPPs) turn out to be a natural way to model such functions of an underlying point process.
		\begin{defn}
			Let $\Phi = \sum_i \varepsilon_{X_i}$ be a Poisson point process on $\R^2$. With each point $X_i$ of $\Phi$, associate a mark $Z_i$ taking values in some measurable space $(\mathbb{K},\mathcal{K})$ such that $\{Z_i\}_{i \in \mathbb{N}} \stackrel{iid}{\sim} \Pi(\cdot)$.
			Then, $\tilde{\Phi} = \sum_i \varepsilon_{(X_i,Z_i)}$ is called an \emph{iid marked point process} on $\R^2 \times \mathbb{K}$ with \emph{mark distribution} $\Pi(\cdot)$.
		\end{defn}
		We now state an ergodic theorem for MPPs which is used to obtain some key results required  in the analysis of the probabilistic forwarding protocol in Section \ref{sec:MPP_fwding}.
		
		\subsection{Ergodic theorem}
		\label{sec:erg_thm}
		Let $(\Omega,\mathcal{F},\P)$ be the probability space over which an iid marked point process $\tilde{\Phi} = \sum_i \varepsilon_{(X_i,Z_i)}$ is defined with mark distribution $\Pi(\cdot)$. Let $\theta_x:\Omega\rightarrow \Omega$, for $x\in \R^2$, be the operator which shifts each point of $\tilde{\Phi}$ by $-x$, i.e., $\theta_x \tilde{\Phi} = \sum_i \varepsilon_{(X_i-x,Z_i)}$ and let $(\mathbb{K},\mathcal{K}) $ be the measurable space of marks. Let $f:\mathbb{K}\times \Omega \rightarrow \R_+$ be a non-negative function of the MPP. Then, by the ergodic theorem for marked random measures (see \cite[Theorem 8.4.4]{BBK}), we have
		\begin{align}
		\frac{1}{\nu(\G_m)}\sum_{X_i \in \G_m}f(Z_i,\theta_{X_i}(\omega)) &\rightarrow \lambda \int_\mathbb{K}\E^{(\0,z)}\left[f(z,\omega) \right]\Pi (dz) \nonumber \\
		&\hspace{3cm} \P\text{-a.s.}
		\label{thm:ergodic}
		\end{align}
		as $m\rightarrow \infty$, where $\E^{(\0,z)}$ is the expectation with respect to the Palm probability $\P^{(\0,z)}$ conditional on the mark,  $z$. If $f(z,\omega) = f(\omega)$, then \eqref{thm:ergodic} reduces to
		\begin{equation}
		\frac{1}{\nu(\G_m)}\sum_{X_i \in \G_m}f(\theta_{X_i}(\omega)) \stackrel{m\rightarrow \infty}{\longrightarrow} \lambda \Palmexp\left[f(\omega)\right] \hspace{0.9cm} \P\text{-a.s.}.
		\label{eq:erg2}
		\end{equation}

		\section{Probabilistic forwarding and MPPs}
		\label{sec:MPP_fwding}
		In this section, we formulate probabilistic forwarding mechanism using the framework of marked point processes. We obtain estimates for $p_{k,n,\delta}$ and $\tau_{k,n,\delta}$ via ergodic theorems for MPPs. It should be noted here that all the graphs and point processes discussed in this section are on the whole $\R^2$ plane.
		\subsection{Single packet probabilistic forwarding}
		\label{sec:singlepkt}
		Consider the probabilistic forwarding of a single packet on $\cG \sim RGG(\Phi,1)$ defined on a Poisson point process (PPP) $\Phi$ of intensity $\lambda$ on $\R^2$. Let $\cG^\0$ be the graph created with the underlying point process being $\Phi^\0 \triangleq \Phi \cup \{\0\}$ as the vertex set, and introducing additional edges from $\0$ to nodes which are within $\Ball$, to the edge set of $\cG$. We assign a mark $1$ to a node if it decides to transmit the packet and $0$ otherwise. Thus, the mark space is $\mathbb{K}=\{0,1\}$ and $\tilde{\Phi}$ is an iid MPP with a $Ber(p)$ mark distribution. Note that the origin, $\0$, has mark $1$ since it always transmits the packet. Also, the subset of nodes which have mark $1$ form a thinned point process of intensity $\lambda p$, and the subset of vertices with mark $0$ form a $\lambda (1-p)  $--thinned process. Denote these by $\Phi^+$ and $\Phi^-$ respectively, and the corresponding RGGs by $\cG^+$ and $\cG^-$. Notice that the set of vertices of $\Phi^+$ which are in the same cluster as the origin are the vertices which receive the packet from the source and transmit it. Thus, the number of vertices in the cluster containing the origin in $\cG^+$ (call this set of nodes $|C^+_\0|$), is the number of transmissions of the packet. 
		
		In addition to the nodes of the cluster containing the origin in $\cG^+$, the nodes of $\cG^-$ which are within distance $1$ from them, also receive the packet. To account for them, we define for any cluster of nodes $S \subset \Phi^+$, the \emph{boundary} of $S$ as 
		\begin{equation*}
		\partial S = \{\v\in \Phi^- | B_1(\v) \cap S \neq \emptyset\},
		\end{equation*}
		and the \emph{extended cluster} of $S$ to be $S^{\text{ext}} = S \cup \partial S$. Then, the receivers are the nodes in $C_\0^{\text{ext}}$. We refer to this as the extended cluster of the origin. 
		
		Our interest is in large networks in which the origin is likely to be in the infinite cluster of $\cG^\0$. Moreover, since we are interested in a large fraction of nodes in the network to be successful receivers, the extended cluster of the origin has to comprise of a significant number of nodes within $\G_m$. In the limit of large $m$, this means that the extended cluster of the origin is the \emph{infinite extended cluster} (IEC), $C^{\text{ext}}$, defined as the extended cluster of $C^+ := C(\Phi^+)$. This also means that the transmitters correspond to the nodes within $\G_m$ of the infinite cluster of $\Phi^+,$ $C^+$. Thus, in the thermodynamic limit, the expected number of vertices in $C_\0 \cap \G_m$  (resp. $C_\0^{\text{ext}} \cap \G_m$) is well-approximated by the expected number of vertices within $\G_m$ of the infinite cluster $C^+$ (resp., of the IEC $C^{\text{ext}}$) for large $m$. We use the ergodic theorem stated in Section \ref{sec:erg_thm} to obtain almost sure results for the fraction of nodes within $\G_m$ of the infinite cluster $C^+$ and the IEC $C^{\text{ext}}$ in terms of the percolation probability $\theta(\lambda)$.

		\subsection{Application of the ergodic theorem}
		\label{sec:appl_erg}
		Specializing the statement in \eqref{thm:ergodic} to the probabilistic forwarding of a single packet where $\mathbb{K}=\{0,1\}$ and the marks are independent, conditional on $\Phi$, with distribution given by $\Pi(1)=1-\Pi(0)=p$, we obtain,
		\begin{align}
		\frac{1}{\nu(\G_m)}\sum_{X_i \in \G_m}f(Z_i,\theta_{X_i}(\omega)) \stackrel{m\rightarrow \infty}{\longrightarrow}  &\lambda p \E^{(\0,1)}[f(1,\omega)] \nonumber\\
		&+\lambda (1-p)\ \E^{(\0,0)}[f(0,\omega)] \nonumber \\
		&\hspace{2cm} \P\text{-a.s.}.
		\label{Eq:single_pkt_erg}
		\end{align}
		
		\par We will now use \eqref{eq:erg2} and \eqref{Eq:single_pkt_erg} to obtain key results which will be used to analyze the probabilistic forwarding of a single packet on $\R^2$. In particular, we substitute different functions $f$ in \eqref{eq:erg2} and \eqref{Eq:single_pkt_erg} to obtain the following results:
		\begin{itemize}
			\item $f(z,\omega)=1$.
			The ergodic theorem in \eqref{eq:erg2}
			results in
			\begin{equation}
			\frac{\Phi(\G_m)}{\nu(\G_m)} \ \ \ \stackrel{m\rightarrow \infty}{\longrightarrow}\ \ \  \lambda \ \ \hspace{1.5cm} \P\text{-a.s..}
			\label{Eq:num_nodes}
			\end{equation}
			As a corollary, taking the reciprocals, we obtain 
			\begin{equation}
			\frac{m^2}{\Phi(\G_m)} \ \ \ \stackrel{m\rightarrow \infty}{\longrightarrow}\ \ \  \frac{1}{\lambda} \ \ \hspace{1.5cm} \P\text{-a.s.},
			\label{Eq:recip}
			\end{equation}
			which holds in our setting since $\lambda >\lambda_c$.
			
			\item $f(z,\omega)=z$.
			Substituting in \eqref{Eq:single_pkt_erg}, we see that the sum on the LHS counts the number of nodes which have mark $1$ in $\G_m$. Indeed, we obtain
			\begin{equation}
			\frac{\Phi^+(\G_m)}{\nu(\G_m)} \stackrel{m\rightarrow \infty}{\longrightarrow}  \lambda p \hspace{2.1cm} \P\text{-a.s..}
			\label{thinned_erg}
			\end{equation}
			
			\item Let $C$ be the unique infinite cluster in $\cG$. Using the ergodic theorem in \eqref{eq:erg2} with $f(z,\omega)= \mathds{1}\{\0 \in C\} $, we see that the sum on the LHS counts the number of vertices of $\Phi$ which are present in the infinite cluster. Then, we have that
			\begin{equation}
			\frac{|C \cap \G_m|}{\nu(\G_m)} \stackrel{m\rightarrow \infty}{\longrightarrow}\lambda \ \theta(\lambda) \hspace{1.6cm}\P\text{-a.s.}.
			\label{Eq:erg_cluster}
			\end{equation}
			Using the dominated convergence theorem (DCT) and \eqref{Eq:recip}, we also have that
			\begin{equation}
			\E \left[\frac{|C \cap \G_m|}{\Phi(\G_m)}\right] \stackrel{m\rightarrow \infty}{\longrightarrow} \ \theta(\lambda).
			\label{eq:perc_prob}
			\end{equation}
			This means that, for large $m$, the expected fraction of vertices of the infinite cluster within $\G_m$ is a good approximation for the percolation probability. We use this to obtain an empirical estimate of the percolation probability as follows. We generate $100$ instantiations of the $RGG(\lambda,1)$ model on $\G_{251}$, for each value of $\lambda$ between $1$ and $5$ (in steps of $0.01$). The average number of vertices in the largest cluster within $\G_{251}$ is computed and taken as a proxy for the fraction of nodes of the infinite cluster. The graph obtained is shown in Figure \ref{fig:theta_lambda}. We use the values from this plot in our numerical results. 
			\begin{figure}
				\centering
				\includegraphics[width=0.5\textwidth]{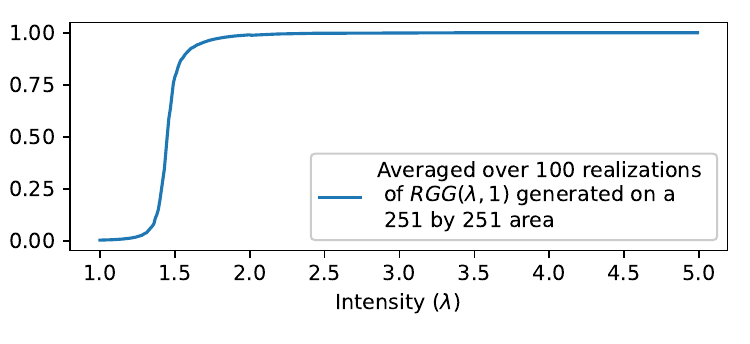}
				\caption{Percolation probability $\theta(\lambda)$ vs. intensity $\lambda$}
				\label{fig:theta_lambda}
			\end{figure}%
			
			\item Suppose $\lambda p>\lambda_c$, so that $\cG^+$ operates in the super-critical region. Let $C^+$ be the unique infinite cluster in $\cG^+$. Since $\Phi^+$ is a thinned point process of intensity $\lambda p$, we can use the result from \eqref{Eq:erg_cluster} for the infinite cluster $C^+$ to obtain 
			\begin{equation}
			\frac{|C^+\cap \G_m|}{\nu(\G_m)} \ \ \ \stackrel{m\rightarrow \infty}{\longrightarrow} \ \ \ \lambda p \ \theta(\lambda p)\ \ \ \ \ \ \ \P\text{-a.s.}.
			\label{Eq:erg_inf}
			\end{equation}
			
			\item Suppose that $\lambda p>\lambda_c$ and let $C^{\text{ext}}$ denote the extended cluster of $C^+$, i.e. $C^{\text{ext}} = C^+ \cup \partial C^+$. Note that since $C^+$ is infinite, $C^{\text{ext}}$ is also infinite. Hence, we refer to it as the \emph{infinite extended cluster}, or IEC for short.  Take $f(\omega) = \mathds{1}(\Ball \cap C(\Phi^+) \neq \emptyset)$.  Observe that $\{X_i \in C^{\text{ext}}\} =  \mathds{1}(B_1(X_i) \cap C(\Phi^+) \neq \emptyset) = f(\theta_{X_i}w)$.  So, using \eqref{eq:erg2},  we have that  
			\begin{align*}
			\frac{1}{\nu(\G_m)} \sum_{X_i \in \G_m} \mathds{1}\{X_i \in& C^{\text{ext}}\} \ \\
			&\stackrel{m\rightarrow \infty}{\longrightarrow}\lambda \P(\Ball \cap C(\Phi^+) \neq \emptyset)\\ &  \hspace{4cm} \P\text{-a.s.}..
			\end{align*}
			By definition, $\P(\Ball \cap C(\Phi^+) \neq \emptyset) =  \theta(\lambda p)$,  the percolation probability of $\Phi^+$. We then have,
			\begin{align}
			\frac{|C^{\text{ext}} \cap \G_m|}{\nu(\G_m)} & \stackrel{m\rightarrow \infty}{\longrightarrow} 
			\lambda \theta (\lambda p) \hspace{2cm} \P\text{-a.s.}
			\label{eq:erg_iec}
			\end{align}
			Thus, it is natural to define,  $\theta^{\text{ext}} (\lambda , p) := \Palm(\0 \in C^{\text{ext}}) = \theta (\lambda p)$.

			Comparing RHS of \eqref{eq:erg_iec} and \eqref{thinned_erg} suggests an alternate viewpoint for the nodes that are present in the IEC. On the underlying point process $\Phi$, define new iid marks $Z' \in \mathbb{K} =\{0,1\}$ with $\mathrm{Ber}(\theta^{\text{ext}}(\lambda,p))$ distribution. This means that a vertex is attributed mark $1$, if it is in the IEC when probabilistic forwarding is carried out with forwarding probability $p$. Then, the fraction of nodes in the IEC when marks are $Z$ corresponds to the fraction of nodes with mark $1$ when marks are $Z'$. This interpretation will be useful in proposing a heuristic approach for probabilistic forwarding of multiple packets in Section \ref{sec:heuristic}.
		\end{itemize}
		
		\subsection{Probabilistic forwarding of multiple packets}
		\label{sec:MPP_mult}
		Consider now the probabilistic forwarding mechanism on $n$ packets. Each node transmits a newly received packet with probability $p$ independently of other packets. It is required to find the fraction of successful receivers, the nodes that receive at least $k$ out of the $n$ packets. From our discussion of probabilistic forwarding of a single packet (in Section \ref{sec:singlepkt}), for large $m$, the number of nodes within $\G_m$ that receive a packet from the origin is well-approximated by the number of nodes in the IEC. In a similar way, the fraction of successful receivers within $\G_m$ can be well approximated by the fraction of nodes which are present in at least $k$ out of the $n$ IECs when probabilistic forwarding is done on the RGG, $\cG^\0$. In this subsection, we will use the ergodic theorem and obtain explicit bounds on this fraction.
		\par Equip each vertex of the point process $\Phi$ with mark $\bZ=(Z_1,Z_2,\cdots,Z_n) \in \mathbb{K} = \{0,1\}^n$. Here the $j$-th co-ordinate of the mark represents transmission of the $j$-th packet on $\Phi$. More precisely, $Z_j(\cdot) \sim Ber(p)$ and, for two different vertices $u$ and $v$,  $\bZ(X_u)$ and $\bZ(X_v)$ are independent conditional on $\Phi$. Therefore, it forms an iid marked point process.  Define $C_{k,n}^{\text{ext}}$ to be the set of nodes which are present in at least $k$ out of the $n$ IECs. Taking $f(z,\omega) = \1\{\0 \in C_{k,n}^{\text{ext}}\}$ in the statement of the ergodic theorem, we obtain
		\begin{equation*}
		\frac{1}{\nu(\G_m)} \sum_{X_i \in \G_m} \mathds{1}\{X_i \in C_{k,n}^{\text{ext}}\}\ \stackrel{m\rightarrow \infty}{\longrightarrow} \lambda \  \Palm(\0 \in C_{k,n}^{\text{ext}})\ \  \P\text{-a.s.}.
		\end{equation*}
		Denote by $\theta_{k,n}^{\text{ext}}(\lambda,p) := \Palm(\0 \in C_{k,n}^{\text{ext}})$. Then the above statement reads as
		\begin{equation}
		\lim_{m\rightarrow \infty} \frac{| C_{k,n}^{\text{ext}}  \cap \G_m|}{\nu(\G_m)}\  = \  \lambda \  \theta_{k,n}^{\text{ext}}(\lambda,p)\hspace{1cm} \P\text{-a.s.}.
		\label{succ_recs}
		\end{equation}
		
		\section{Main results}
		\label{sec:results}
		In this section, we will obtain expressions for the expected fraction of successful receivers and the expected total number of transmissions on the finite graph $G$ based on the framework that has been developed in the previous section. 
		
		While constructing $\cG^\0$ (as described in Section \ref{sec:singlepkt}), the graph corresponding to $\G_{0,0}$ can be taken to be $G_m^\0$ (with additional edges from vertices in $\G_{0,0}$ to those outside it). Alternately, $G_m^\0$ can be constructed by considering a restriction of $\cG \sim RGG(\lambda,1)$ to $\G_m$ and connecting the origin to nodes within $\Ball$. In essence, it is true that the  distribution of nodes of $G_m^\0$  and $\cG^\0 \cap \G_m$ is the same. Recall that the graph $G$ on which the probabilistic forwarding mechanism is carried out, is the component of the origin in $G_m^\0$. In light of the correspondence between the vertices of $G_m^\0$ and $\cG^\0 \cap \G_m$, the graph $G$ should correspond to the graph induced on the nodes within $ \G_m$ that are present in the cluster of the origin in $\cG^\0$. However, these nodes also include those that are contained in the cluster of the origin through paths which go outside $\G_m$ but are not connected to the origin within $\G_m$ (see Fig. \ref{fig:ckt}). We refer to these as, nodes in the cluster of the origin but \emph{without a $\G_m$-conduit} and denote them by $\widehat{C}_{\0,m}$. The following theorem states that the number of nodes without $\G_m$-conduits normalized by the area of $\G_m$ converges almost surely to $0$.
		\begin{thm}
			For $\lambda > \lambda_c$,
			$$\displaystyle{\lim\limits_{m\rightarrow \infty} \frac{|\widehat{C}_{\0,m}|}{m^2} = 0 \hspace{1cm} \P\text{-a.s.}}.$$
			As a consequence, we have
			$$\lim\limits_{m\rightarrow \infty} \frac{|C_\0(G_m^\0)|}{\lambda m^2} = \lim\limits_{m\rightarrow \infty} \frac{|C_\0(\cG^\0) \cap \G_m |}{\lambda m^2} \hspace{1cm} \P\text{-a.s.},$$
			where $C_\0(\cG^\0)$ is the set of nodes in the cluster of the origin in $\cG^\0$.
			\label{thm:rgg_out_recs}
		\end{thm}
		The latter part of the theorem is obtained by noting that $C_\0(\cG^\0)\  \cap\  \G_m = C_\0(G_m^\0) \  \cup\  \widehat{C}_{\0,m}$ with $C_\0(G_m^\0) \  \cap \  \widehat{C}_{\0,m} = \emptyset $.
	For the first part, we divide the nodes in $\widehat{C}_{\0,m}$ into those which are present within a smaller concentric $r \times r$ area $\G_{r}$, for $r<m$, and those in $\G_m \setminus \G_{r}$ (see Fig. \ref{fig:ckt}). Denote these by
		\begin{equation*}
		\widehat{S}_{r,m} = \widehat{C}_{\0,m} \cap \G_{r} \hspace{1cm}\text{and}\hspace{1cm} \widehat{T}_{r,m} = \widehat{C}_{\0,m} \setminus \widehat{S}_{r,m}
		\end{equation*}
		respectively. In the following two lemmas, we show that for an appropriate value of $r$, the number of nodes in $\widehat{S}_{r,m}$ and $\widehat{T}_{r,m}$ normalized by $m^2$ converges to $0$ almost surely.
		
		Define $s_m=\frac{m-r}{2}$, the width of the annulus $\G_m \setminus \G_r$. Let us first look at the nodes in $\widehat{T}_{r,m}$. The following lemma states that the fraction of nodes of $\widehat{T}_{r,m}$ in a narrow annulus within $\G_m$ approaches $0$ as $m\rightarrow \infty$. 
		\begin{lem} 
			For a sequence $s_m \rightarrow \infty $ with $\frac{s_m}{m} \rightarrow 0$ as $m\rightarrow \infty$, we have
			$$\lim\limits_{m\rightarrow \infty} \frac{|\widehat{T}_{r,m}|}{m^2} = 0 \hspace{1cm} \P\text{-a.s.}$$
			\label{lem:ann_nodes}
		\end{lem}
		\begin{proof}	
			The nodes in $\widehat{T}_{r,m}$ form a subset of the nodes of the underlying Poisson point process $\Phi$ which are within $\G_m \setminus \G_r$. Thus, we have,
			\begin{equation}
			|\widehat{T}_{r,m}| \le \Phi(\G_m \setminus \G_r) \hspace{1cm} \P\text{-a.s.}
			\label{eq:trm}
			\end{equation}
			It suffices now to show that $\frac{\Phi(\G_m \setminus \G_r)}{m^2} \rightarrow 0$ as $m\rightarrow \infty$, which then proves the lemma. We proceed as follows:
			\begin{align}
			\frac{\Phi(\G_m \setminus \G_r)}{m^2} &= \frac{\Phi(\G_m \setminus \G_r)}{m^2-r^2} \cdot \frac{m^2-r^2}{m^2} \nonumber\\
			&= \frac{\Phi(\G_m \setminus \G_r)}{m^2-r^2} \cdot \left(\frac{4s_m}{m}-\frac{4s_m^2}{m^2} \right).
			\label{eq:term_split}
			\end{align}
			Using the ergodic result in \eqref{Eq:num_nodes} with $\G_m$ replaced by $\G_m \setminus \G_r$, we obtain
			$$\frac{\Phi(\G_m \setminus \G_r)}{m^2-r^2} \rightarrow \lambda \hspace{1cm} \P \text{-a.s.}.$$
			This is because the area of $\G_m\setminus \G_r$ is $m^2-r^2$. Moreover, since the term within parenthesis in \eqref{eq:term_split} converges to $0$, from the condition in the statement of the lemma, we have that
			\begin{align*}
			\lim\limits_{m\rightarrow \infty} \frac{|\widehat{T}_{r,m}|}{m^2} &\le \lim\limits_{m\rightarrow \infty} \frac{\Phi(\G_m \setminus \G_r)}{m^2}\\
			&= 0 \hspace{1cm} \P\text{-a.s.}.
			\end{align*}
			
		\end{proof}

		\begin{figure}
			\centering
			\includegraphics[width=0.45\textwidth]{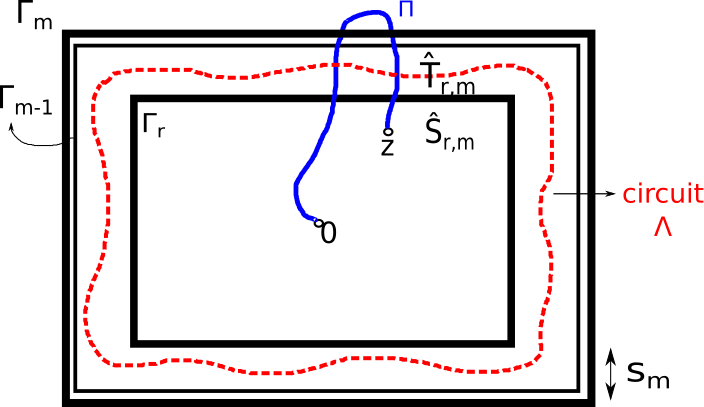}
			\caption{Circuit in the annulus $\G_{m-1} \setminus \G_r$}
			\label{fig:ckt}
		\end{figure}%
		We next address the nodes in $\widehat{S}_{r,m}$. These are nodes within $\G_r$ but without a $\G_m$-conduit. We will show that $|\widehat{S}_{r,m}|$ converges to $0$ almost surely using ideas from Russo-Seymour-Welsh (RSW) theory which is discussed in Appendix \ref{app:RSW}. For this, let $\text{Ann}_{s_m}$ denote the event of existence of a circuit in the annulus $\G_{m-1} \setminus \G_r$ as shown in Fig. \ref{fig:ckt}. Notice that if $|\widehat{S}_{r,m}| > 0$, then there cannot be such a circuit. This is stated formally in the following lemma.
		\begin{lem}
			For $\lambda> \lambda_c$, let $\widehat{B}_{r,m}$ be the event that there exists at least one point of $\widehat{C}_{\0,m}$ within $\G_r$ without a $\G_m$-conduit i.e., $\widehat{B}_{r,m} = \{|\widehat{S}_{r,m}| > 0\}$. Then $\widehat{B}_{r,m} \subseteq \text{Ann}_{s_m}^c$.
			\label{lem:small_box_recs} 
		\end{lem}
		\begin{proof}
		 The proof proceeds by showing that the events $\widehat{B}_{r,m}$ and $\text{Ann}_{s_m}$ cannot occur simultaneously. For this, suppose there is a circuit $\Lambda$ within  $\Gamma_{m-1} \setminus \Gamma_r$.  Also, suppose that some point $z \in \Phi$ that lies within $\Gamma_r$ is connected to the origin only via a path $\Pi$ that leaves $\Gamma_m$.  Then, $\Pi$ must physically cross $\Lambda$ at least twice as shown in Fig. \ref{fig:ckt}. At any of the locations where such a crossing happens, consider the two adjacent points, $x$ and $y$, of $\Phi$ that are on the path $\Pi$, but which fall on opposite sides of $\Lambda$. Note that, since $\Lambda$ is at a distance of at least $1$ from the boundary of $\Gamma_m$, both $x$ and $y$  are within $\Gamma_m$.  Also consider the two adjacent points, $u$ and $v$, of $\Phi$ that are on $\Lambda$, but which fall on opposite sides of $\Pi$. Now, $x,u,y,v$  form a quadrilateral with diagonals $xy$ and $ uv $ having length at most $1$. Hence, at least one of the sides of this quadrilateral has length at most $1$. This means that at least one of $x$ and $y$ is within distance $1$ of either $u$ or $v$ (or both).  Thus, at any crossing of $\Pi$ and $\Lambda$, either $\Pi$ and $\Lambda$ intersect at some point of $\Phi$, or $\Pi$ is connected by an edge to the circuit $\Lambda$, and the connecting edge lies entirely within $\Gamma_m$.  From this, one can construct a $\Gamma_m$-conduit between $z$ and the origin.
		\end{proof}
		
		\begin{cor}
			For $\lambda> \lambda_c$, there exists $s_m \ll m $ such that $$|\widehat{S}_{r,m}| \stackrel{m\rightarrow \infty}{\longrightarrow} 0 \hspace{1cm}\P\text{-a.s.}.$$
			\label{corr:small_box_recs}
		\end{cor}
		\begin{proof}
			 Let $\epsilon>0$. From the previous lemma and using Proposition \ref{prop:annulus}, we can  write
			\begin{equation}
			\P(|\widehat{S}_{r,m}|>\epsilon) \le \P(\text{Ann}_{s_m}^\mathsf{c})  \le 8 \left\lceil \frac{m}{s_m}\right\rceil \exp(-cs_m).
			\end{equation}	
			Taking $s_m = \frac{3\log m}{c}$ and summing over $m$, we obtain
			\begin{equation}
			\sum_m \P(|\widehat{S}_{r,m}|>\epsilon) \le \sum_{m}\frac{c'}{m^2\log (m)}+\frac{c''}{m^3} <\infty.
			\label{eq:in_nodes}
			\end{equation}
			Using the Borel-Cantelli lemma, this shows that $|\widehat{S}_{r,m}| \rightarrow 0$ as $m\rightarrow \infty$ almost surely.
		\end{proof} 
		
		\begin{proof}[Proof of Theorem \ref{thm:rgg_out_recs}]
			The choice of $s_m = \frac{3\log m}{c}$ satisfies the condition of Lemma \ref{lem:ann_nodes} as well. From Lemma \ref{lem:ann_nodes} and Corollary \ref{corr:small_box_recs},  as $m$ tends to infinity, we obtain
			\begin{equation*}
			\frac{\widehat{C}_{\0,m}}{m^2} =\frac{|\widehat{T}_{r,m}|}{m^2}+\frac{|\widehat{S}_{r,m}|}{m^2} \rightarrow 0 \hspace{1cm} \P\text{-a.s.},
			\end{equation*}
			where $r=m-\frac{6\log m}{c}$. This proves the theorem.
		\end{proof}

		Continuing the discussion prior to Theorem \ref{thm:rgg_out_recs}, the fraction of nodes in the component of the origin that are not connected via $\G_m$-conduits approaches $0$ as $m\rightarrow \infty$ almost surely. The outcome of Theorem \ref{thm:rgg_out_recs} is that in the asymptotic regime as $m\rightarrow \infty$, as long as we are interested in the fraction of nodes within the component of the origin, it does not matter whether these are connected to the origin via $\G_m$-conduits or not. In other words, the fraction of nodes within $G$ can be approximated by the fraction of nodes within $\G_m$ of the component of the origin in $\cG^\0$ for a large $m$. To get a handle on the fraction of nodes within $\G_m$ of $C_\0(\cG^\0)$, we will need the following lemma.
		\begin{lem}Let $A = \{\0 \in C(\cG^\0) \}$, where $C(\cG^\0)$ is the infinite cluster of $\cG^\0$. For $\lambda > \lambda_c$, we then have
			$$\lim\limits_{m\rightarrow \infty }\frac{|C_\0(\cG^\0) \cap \G_m|}{\lambda m^2}\  = \  \theta(\lambda) \1_A \hspace{1cm} \P\text{-a.s.}.$$
			\label{lem:comp_org}
		\end{lem}
		\begin{proof}
			We can write $$\frac{|C_\0(\cG^\0) \cap \G_m|}{\lambda m^2} = \frac{|C_\0(\cG^\0) \cap \G_m|}{\lambda m^2} \1_A + \frac{|C_\0(\cG^\0) \cap \G_m|}{\lambda m^2} \1_{A^c}.$$
			Since $A^c$ is the event that the origin is in some finite cluster, the number of nodes within $C_\0(\cG^\0)$ is finite. In the limit as $m\rightarrow \infty$, the latter term on the RHS above goes to $0$. For the first term, notice that $A=\{C_\0(\cG^\0) = C(\cG^\0) \}$. 
			This gives
			$$ \frac{|C_\0(\cG^\0) \cap \G_m|}{\lambda m^2} \1_A =  \frac{|C(\cG^\0) \cap \G_m|}{\lambda m^2} \1_A .$$
			Further, from \eqref{eq:palm_to_normal}, we have that
			$$ \lim\limits_{m\rightarrow \infty } \frac{|C(\cG^\0) \cap \G_m|}{\lambda m^2} =  \lim\limits_{m\rightarrow \infty } \frac{|C(\cG) \cap \G_m|}{\lambda m^2}  \hspace{1cm} \P\text{-a.s.}.$$
			Therefore, using \eqref{Eq:erg_cluster} in the RHS of the above equation, we obtain that 
			\begin{align*}
			\lim\limits_{m\rightarrow \infty } \frac{|C_\0(\cG^\0) \cap \G_m|}{\lambda m^2} \1_A &= \lim\limits_{m\rightarrow \infty } \frac{|C(\cG^\0) \cap \G_m|}{\lambda m^2} \1_A \\
			&= \theta(\lambda)\1_A \hspace{1cm} \P\text{-a.s.}.
			\end{align*}
		\end{proof}
		
		\textbf{Note}: It should be noted here that the statements in Theorem \ref{thm:rgg_out_recs}, Lemmas \ref{lem:small_box_recs} and \ref{lem:comp_org} and Corollary \ref{corr:small_box_recs} hold $\Palm$-a.s., since these are $\P$-a.s. statements made on the underlying graph $\cG^\0$.
		
		Before we proceed, we recall the definition of the minimum forwarding probability in \eqref{Eq:pkndelta}:
		\begin{equation*}
		p_{k,n,\delta} = \inf \left\{p \ \ \Big{|} \ \ \mathbb{E} \left[ \frac{R_{k,n}(G_m^\0)}{|C_\0(G_m^\0)|} \right] \geq 1-\delta \right\},
		\end{equation*}
		where the expectation is over the graph as well as the probabilistic forwarding mechanism.  Note that in our setting, the source, $\0$, always has mark $1$ since it transmits all the $n$ packets. To be more explicit, define $\1=(1,1,\cdots,1)$ to be the vector of all $1$s of length $n$. We denote by $\MarkPalm$ the expectation with respect to the Palm probability $\P^{\0}$ given a point at the origin, conditional on it having mark $\bZ(\0)=\1$. In terms of this, the above equation translates to
		\begin{equation}
		p_{k,n,\delta} = \inf \left\{p \ \ \Big{|} \ \ \MarkPalm \left[ \frac{R_{k,n}(G_m)}{|C_\0(G_m)|} \right] \geq 1-\delta \right\}.
		\label{pkndelta}
		\end{equation}
		
		Next, since we are addressing a broadcast problem, it is necessary that a large fraction of nodes receive a packet. This, in turn necessitates that the fraction of nodes that transmit the packet is also large. With reference to the RGG on the whole plane, this means that the nodes in $\cG^+$ need to have an infinite cluster. To allow for this, we make the following assumption.
		\begin{Assumption}
			The forwarding probability $p$ is such that $\lambda p > \lambda_c$.
			\label{assump:super_crit}
		\end{Assumption}
		Notice that the $p_{k,n,\delta}$ values obtained from simulations in Figure \ref{fig:RGG_simu} conform to this assumption. The assumption is discussed in slightly more detail in Section \ref{sec:super}. We now obtain expressions for the minimum forwarding probability and the expected total number of transmissions based on these two assumptions.
		
		\subsection{Transmissions}
		Consider first the transmission of a single packet. Let $T(G_m)$ be the number of nodes of $G_m$ that receive the packet from the source and transmit it and let $\mathcal{T}(\cG) \cap \G_m$ be the set of nodes within $\G_m$ that receive the packet from the source and transmit it when probabilistic forwarding is carried out on $\cG$ \footnote{\footnotesize{It is implicit from the use of Palm probabilities that the origin is the source and probabilistic forwarding is formulated as an MPP as described in Section \ref{sec:singlepkt}.}}. From our construction, it follows that $T(G_m)$ is stochastically dominated by $|\mathcal{T}(\cG) \cap \G_m |$ since there might be nodes which receive a packet from outside $\G_m$ and transmit it. However, it can be shown that, 
			\begin{equation*}
			\lim_{m\rightarrow \infty} \frac{\E^{(\0,1)}\left[T(G_m)\right]}{m^2} = \lim_{m\rightarrow \infty} \frac{\E^{(\0,1)}\left[|\mathcal{T}(\cG) \cap \G_m |\right]}{m^2}.
			\end{equation*}
			This is because the expected fraction of transmitting nodes with no $\G_m$-conduits diminishes as $m \rightarrow \infty$. Thus, it suffices to evaluate $\lim_{m\rightarrow \infty}\frac{\E^{(\0,1)}\left[|\mathcal{T}(\cG) \cap \G_m |\right]}{m^2}$ to find the expected number of transmissions for a single packet.
		
		In the jargon of marked point processes, $\mathcal{T}(\cG)$ is the set of vertices with mark $Z(\cdot)=1$ that are in the cluster containing the origin. Note that the origin has mark $1$, since it always transmits the packet. As the vertices with mark $1$ form a thinned point process, $\Phi^+$ of intensity $\lambda p$, $\mathcal{T}(\cG)$ is the set of nodes in the cluster containing the origin in $\cG^+$. In Section \ref{sec:singlepkt}, we denoted this set by $C^+_\0$. From Assumption \ref{assump:super_crit}, the graph on $\Phi^+$ is in the super-critical regime and thus possesses a unique infinite cluster, $C^+$. The following theorem provides the expected size of $C^+_\0 \cap \G_m$. The proof proceeds by relating it to the expected size of $C^+ \cap \G_m$ and using the ergodic result in \eqref{Eq:erg_inf}.
		\begin{thm}
			For $\lambda p>\lambda_c$, we have
			\begin{equation*}
			\lim_{m\rightarrow \infty} \E^{(\0,1)} \left[\frac{|C^+_\0 \cap \G_m|}{\lambda m^2} \right] =  p \, \theta(\lambda p)^2 .
			\end{equation*}
			\label{thm:trans}
		\end{thm}
			\begin{proof}
				Denote by $C^+$, the unique infinite cluster of the thinned process $\Phi^+$. Define the event $A^+=\{ \0 \in C^+\} = \{B_1(\0) \cap C^+ \neq \emptyset\} \cap \{Z(\0) = 1\}$. 
				Using Lemma \ref{lem:comp_org} for the thinned process $\Phi^+$ of intensity $\lambda p$, we obtain
				\[ \frac{|C^+_\0 \cap \G_m|}{\lambda p m^2} \ \  \stackrel{m\rightarrow \infty}{\longrightarrow}  \ \ \theta(\lambda p) \mathds{1}\{A^+\}  \hspace{2cm} \P\text{-a.s.}. \]
				From the note following Lemma \ref{lem:comp_org} and using DCT, the expected values with respect to $\Palm$ also converge giving,
				\begin{align*}
				\lim_{m\rightarrow \infty} \E^{\0} \left[\frac{|C^+_\0 \cap \G_m|}{\lambda m^2} \right] & =  \ \ \theta(\lambda p) \Palm(A^+) \nonumber \\
				& = p \, \theta(\lambda p)^2,
				\end{align*}
				where the last equality uses the definition of $A^+$, $\Palm(A^+) = \Palm(B_1(\0) \cap C^+ \neq \emptyset)\ \Palm(Z(\0)=1) = p\theta(\lambda p)$, and we have also used that $\{B_1(\0) \cap C^+ \neq \emptyset\}$ and $\{Z(\0)=1\}$ are independent events with respect to $\Palm$. The proof is complete by noting that if $Z(\0) = 0$,  then $C^+_\0 = \emptyset$ and so
				$$ \E^{\0} \left[\frac{|C^+_\0 \cap \G_m|}{\lambda m^2}  \right] = p \,  \E^{(\0,1)} \left[\frac{|C^+_\0 \cap \G_m|}{\lambda m^2}  \right].$$
			\end{proof}
			
			Therefore, for large values of $m$, the expected number of transmissions, $\E^{\0,1}\left[T(G_m)\right]$, can be approximated by
			\begin{equation*}
			\E^{(\0,1)}\left[|C^+_\0 \cap \G_m|\right] \approx m^2 \lambda p \ \theta(\lambda p)^2 .
			\end{equation*} 
			
			Consider now the transmission of multiple packets. The $n$ coded packets are transmitted independently of each other. The expected total number of transmissions of all $n$ packets would just be $n$ times the expected transmissions of a single packet. Therefore, from Theorem \ref{thm:trans}, we then obtain
			\begin{equation}
			\tau_{k,n,\delta} \approx n m^2 \lambda p_{k,n,\delta} \ \left(\theta(\lambda p_{k,n,\delta})\right)^2 .
			\label{taukndelta_expr}
			\end{equation}
			
			\subsection{Minimum forwarding probability}
			In this section, we will obtain an expression for the minimum forwarding probability. Recall that this entails estimating $\MarkPalm \left[ \frac{R_{k,n}(G_m)}{|C_\0(G_m)|} \right] $, where $C_\0(G_m)$ is the set of nodes in the component of the origin in the underlying RGG on $\G_m$ and $R_{k,n}(G_m)$ are the number of nodes that receive at least $k$ out of the $n$ packets from the origin, which is the source. From Theorem \ref{thm:rgg_out_recs}, $C_\0(G_m)$ can be viewed as the set of nodes in the component of the origin in $\cG^\0$ restricted to $\G_m$ but with only those nodes which are connected to the origin via $\G_m$-conduits. $R_{k,n}(G_m)$ is the number of nodes among those in $C_\0(G_m)$, which are successful receivers. These arguments lets us think of the expectation $\MarkPalm \left[ \frac{R_{k,n}(G_m)}{|C_\0(G_m)|} \right] $, with respect to the RGG, $\cG^\0$, instead of the finite RGG, $G^\0_m$.
			
			Since we are interested in large networks, it is natural to assume that the origin is part of the infinite cluster of $\cG^\0$. This means that the cluster of the origin in $G_m^\0$ connects to the infinite cluster in $\cG^\0$ when $G_m^\0$ is embedded within it. In other words, the event $A = \{\0 \in C(\cG^\0) \}$ occurs. The results of this section are made with this assumption, which is stated below explicitly. Additional justification for this is provided in Section \ref{sec:super}.
			\begin{Assumption}
				\label{assump:orig}
				The origin is part of the infinite cluster of $\cG^\0$.
			\end{Assumption}
			
			From the discussion above and the assumption, our interest now is to estimate $\MarkPalm_A \left[ \frac{R_{k,n}(G_m)}{|C_\0(G_m)|} \right] $. The subscript $A$ in the expectation $\MarkPalm_A$ indicates conditional expectation given that the event $A$ occurs. From Assumption \ref{assump:super_crit}, it is clear that such a conditioning can indeed be done, since $\P(A) = \theta(\lambda)>0$.
			
			The following theorem gives the expected value of the fraction of successful receivers in the limit as $m\rightarrow \infty$ given the event $A$. Before we state the theorem, recall the formulation of probabilistic forwarding as a marked point process in Section \ref{sec:MPP_fwding}. $C_{k,n}^{\text{ext}}$ was defined as the set of nodes which are present in at least $k$ out of the $n$ IECs and let $\theta_{k,n}^{\text{ext}} \equiv \theta_{k,n}^{\text{ext}}(\lambda,p) = \Palm(\0 \in C_{k,n}^{\text{ext}})$.  Additionally, define $A_{[t]}^{\text{ext}}$ to be the event that the origin is present only in the IECs corresponding to the packets $1,2,\cdots,t$. 
			
			\begin{thm}
				For $\lambda p >\lambda_c$, we have \
				\begin{equation*}
				\lim_{m\rightarrow \infty} \MarkPalm_A  \left[ \frac{R_{k,n}(G_m)}{|C_\0(G_m)|} \right] = 
				\ \frac{1}{\theta(\lambda)^2} \sum_{t=k}^n \binom{n}{t} \theta_{k,t}^{\text{ext}}\  \ProbMark (A_{[t]}^{\text{ext}}).
				\label{thm:recs}
				\end{equation*}
			\end{thm}
			The proof is on similar lines as that on the grid in \cite{ToN}. It relies on carefully relating the fraction of successful receivers on $G$ to the fraction of nodes present in at least $k$ out of the $n$ IECs corresponding to probabilistic forwarding on $\cG^\0$. An outline of the proof is given in Appendix \ref{app:proof_recs}.
			
			The following proposition is used to express $\ProbMark (A_{[t]}^{\text{ext}})$ in terms of  $\theta_{k,n}^{\text{ext}}$. 
			\begin{prop}
				
				\begin{equation}
				\ProbMark \left(A_{[t]}^{\text{ext}} \right) = 
				\begin{cases}
				\displaystyle{\frac{\theta_{t,n}^{\text{ext}}-\theta_{t+1,n}^{\text{ext}}}{\binom{n}{t}}} & 0 \le t \le n-1 \\
				\theta_{n,n}^{\text{ext}} & t=n
				\end{cases}.
				\end{equation}
				
				\label{prop:Atext}
			\end{prop} 
			\begin{proof}
				The second part follows directly from the definitions of $\theta_{n,n}^{\text{ext}}$ and the event $A_{[n]}^{\text{ext}}$. For the first part, define for $T \subseteq [n]$, $A_T^{\text{ext}}$ to be the event that the origin is present in exactly the IECs indexed by $T$. Note that 
				$$	\theta_{k,n}^{\text{ext}} = \ProbMark(\0 \in C_{k,n}^{\text{ext}}) \\
				=\sum_{j=k}^{n}\sum_{\substack{T\subseteq [n] \\ |T|=j}}\ProbMark(A_T^{\text{ext}}).
				$$
				Since the event $A_T^{\text{ext}}$ depends only on the cardinality $j$ (see Step 7 in Appendix \ref{app:proof_recs}), we obtain 
				$$\theta_{k,n}^{\text{ext}} = \sum_{j=k}^{n}\binom{n}{j}\ProbMark(A_{[j]}^{\text{ext}}). $$
				We then have that $
				\theta_{t,n}^{\text{ext}}-\theta_{t+1,n}^{\text{ext}} = \binom{n}{t}\ProbMark(A_{[t]}^{\text{ext}})$ for $0 \le t \le n-1$, which is the statement of the proposition.
			\end{proof}
			
			We remark here that the statement of Theorem \ref{thm:recs} can be used to obtain an estimate for the expected fraction of successful receivers without the conditioning on the event $A$. We write 
			\begin{align*}
			\MarkPalm  \left[ \frac{R_{k,n}(G_m)}{|C_\0(G_m)|} \right]\  &= \  \theta(\lambda) \  \MarkPalm_A  \left[ \frac{R_{k,n}(G_m)}{|C_\0(G_m)|} \right] + \\
			& \hspace{1cm} (1-\theta(\lambda))\  \MarkPalm_{A^C}  \left[ \frac{R_{k,n}(G_m)}{|C_\0(G_m)|} \right]
			\end{align*}
			Notice from Fig. \ref{fig:theta_lambda} that $\theta(\lambda)$ shows a phase transition phenomenon. For the intensities we are interested in, $\P(A^c) = 1-\theta(\lambda)$ is very small and the latter term in the above equation can be neglected. This also suggests that Assumption \ref{assump:orig} is not a very strong requirement.
			
			Consequently, for large $m$, using Theorem \ref{thm:recs} and Proposition \ref{prop:Atext} in \eqref{pkndelta} yields an approximation for the minimum forwarding probability given by,
			\begin{equation}
			\displaystyle{p_{k,n,\delta} \approx \inf \left\{p \Bigg| \sum_{t=k}^{n-1} \frac{\theta_{k,t}^{\text{ext}}(\theta_{t,n}^{\text{ext}}-\theta_{t+1,n}^{\text{ext}})}{\theta(\lambda)}+\frac{\theta_{k,n}^{\text{ext}}\theta_{n,n}^{\text{ext}}}{\theta(\lambda)}   \geq 1-\delta \right\}.}
			\label{pkndelta_expr}
			\end{equation}
			
			\subsection{Comparison with simulations}
			We have not been able to obtain exact expressions for the probability $\theta_{k,t}^{\text{ext}}(\lambda,p)$ in terms of the percolation probability $\theta(\lambda)$. However, in Section \ref{sec:bounds}, we provide some bounds for it. We also develop an alternate heuristic approach, which provides comparable results for the minimum forwarding probability obtained through simulations, in Section \ref{sec:heuristic}.
			
			Nevertheless, the approximation for the expected total number of transmissions, $\tau_{k,n,\delta}$ in \eqref{taukndelta_expr} can be evaluated with the knowledge of the minimum forwarding probability. In Fig. \ref{fig:comp_theory_simu}, we show the plot of $\tau_{k,n,\delta}$ normalized by $\lambda m^2$ with $n$ in which we use $p_{k,n,\delta}$ values from Fig. \ref{fig:RGG_simu}(a)
			
			\begin{figure}
				\centering
				\includegraphics[width=0.5\textwidth]{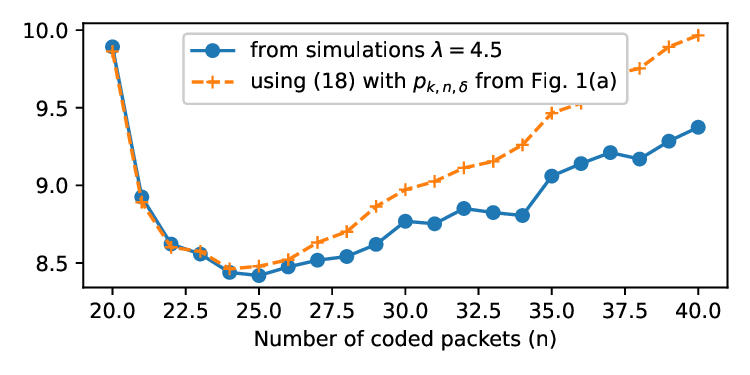}
				\caption{Comparison of the expected number of transmissions per node in the RGG($4.5,1$) model on $\G_{101}$ obtained using \eqref{taukndelta_expr} with that obtained through simulations. Note that the $p_{k,n,\delta}$ value for each point on both the curves are from the simulations in Fig. 1(a).}
				\label{fig:comp_theory_simu}
				\vspace{-1.5em}
			\end{figure}%
			It is observed that for $n \lesssim 26$, both the curves match pretty well. However, for $n>26$ they diverge. This can be attributed to the fact that as $n$ increases, $p_{k,n,\delta}$ decreases as in Fig \ref{fig:RGG_simu}(a) and thus $\lambda p_{k,n,\delta} \searrow \lambda_c$. The estimate for the percolation probability, $\theta(\lambda)$, obtained via the ergodic result in \eqref{eq:perc_prob} may not be accurate near the critical intensity, $\lambda_c$ (which is itself not exactly known). In particular, $\G_{251}$ may not be large enough for the ergodic result in \eqref{eq:perc_prob} to kick in, as we approach $\lambda_c$.
			
			Nevertheless, this provides justification to our observation that the expected number of transmissions indeed decreases when we introduce coded packets along with probabilistic forwarding. This comes with a catch that the minimum forwarding probability for a near-broadcast behaves as in Fig \ref{fig:RGG_simu}(a). In order to establish this, we provide a heuristic explanation for it in the next section.

			\section{A heuristic argument}
			\label{sec:heuristic}
			In the marked point process formulation, probabilistic forwarding of multiple packets was modeled using marks given by $\bZ=(Z_1,Z_2,\cdots,Z_n)$ with $Z_i \sim Ber(p)$ on the underlying point process $\Phi$. We refer to this as the \emph{original model}. Motivated by the alternate interpretation for the nodes in the IEC expounded at the end of Section \ref{sec:appl_erg}, in this section, we provide a heuristic approach for evaluating the minimum forwarding probability. 
			
			As before, let $\theta^{\text{ext}}(\lambda,p)$ denote the probability that the origin is in the IEC for a single packet transmission. Associate a new mark $\bZ'=(Z'_1,Z'_2,\cdots,Z'_n) \in \mathbb{K} = \{0,1\}^n$ to each vertex of $\Phi$. The $i$-th co-ordinate of $\bZ'$ corresponds to probabilistic forwarding of the $i$-th packet. The mark $\bZ'$ is chosen such that each of the $i$ co-ordinates is either $1$ with probability $\theta^{\text{ext}}(\lambda,p)\ ( = \theta(\lambda p))$ or $0$ with the remaining probability, independent of the others. Similar to the viewpoint for the single packet transmission, our idea is to use $Z'_i$ as a proxy for a vertex to be present in the IEC in probabilistic forwarding of the $i-$th packet. We refer to this as the \emph{mean-field model.}
			
			There are two key differences between the two models defined here. Firstly, in the original model, presence of a node in the IEC is not independent of other nodes being present in the IEC. Whereas, in the mean-field model, $Z'_i(\u)$ and $Z'_i(\v)$ are chosen to be independent $Ber(\theta(\lambda p))$ random variables for two distinct vertices $\u$ and $\v$. Since $Z'_i$ is interpreted as an indicator whether a vertex is present in the $i$-th IEC, this independence is enforced, conditional on $\Phi$. Secondly, in the original model,  presence of a particular node in IECs corresponding to two different packets, are not independent. They are independent conditional on $\Phi$ but not otherwise. In the mean-field model, since $\bZ'_i(\v)$ and $\bZ'_j(\v)$ are taken to be iid, this dependence is over-looked. 
			
			To analyze the mean-field model, let us use the ergodic theorem \eqref{thm:ergodic} with
			\begin{equation*}
			f(\bz',\omega)= \sum_{j=k}^n \sum_{\substack{T \subseteq [n] \\ |T|=j}} \prod_{i \in T}z'_i \prod_{i \notin T} (1-z'_i).
			\end{equation*}
			The inner summation is $1$ only if a node has mark $1$ in exactly the co-ordinates indexed by $T$ (which has cardinality $j$). Since the outer sum goes over all $j \geq k$, the value of the function is $1$ for a vertex which has mark $1$, in at least $k$ out of the $n$ co-ordinates. From our interpretation of $\bZ'$, the value of the function, $f$, for a vertex is equal to $1$ if it is present in at least $k$ out of the $n$ IECs of the original model. Define $C'_{k,n}$ to be the set of nodes which have mark $Z'_i(\cdot) = 1$ in at least $k$ out of the $n$ packet transmissions in the mean-field model. Here, $C'_{k,n}$ acts as a proxy for $C_{k,n}^{\text{ext}}$. Since $f(\bZ'(\v),\omega)=1$ if $\v \in C'_{k,n}$, we can apply Theorem \ref{thm:ergodic}, to obtain for $\P$ almost surely
			\begin{align*}
			\frac{1}{\nu(\G_m)} &\sum_{X_i \in \G_m} \mathds{1}\{X_i \in C'_{k,n}\}\\
			& \stackrel{m\rightarrow \infty}{\longrightarrow} \lambda \sum_{\bz' \in \{0,1\}^n} \P(\bZ'=\bz') \\
			&\hspace{2cm} \E^{(\0,\bz')}\left[\sum_{j=k}^n \sum_{\substack{T \subseteq [n] \\ |T|=j}} \prod_{i \in T}\bZ'_i \prod_{i \notin T} (1-\bZ'_i)\right] \\
			&= \lambda \sum_{j=k}^n \sum_{\substack{T \subseteq [n] \\ |T|=j}} \sum_{\bz \in \{0,1\}^n} \P(\bZ'=\bz') \times \prod_{i \in T}\bz'_i \prod_{i \notin T} (1-\bz'_i).
			\end{align*}
			For a fixed $j$ and a set $T$ with $|T|=j$,  there is exactly one $\bz'$ such that  $\prod_{i \in T}\bz'_i \prod_{i \notin T} (1-\bz'_i) = 1$ and the probability of such a $\bz'$ is given by $\P(\bZ'=\bz') = \theta^{\text{ext}}(\lambda,p)^j \times (1-\theta^{\text{ext}}(\lambda,p))^{n-j} $. Thus, the expression above reduces to
			\begin{align*}
			\frac{| C'_{k,n}  \cap \G_m|}{\nu(\G_m)}
			&	 \stackrel{m\rightarrow \infty}{\longrightarrow}
			\lambda \sum_{j=k}^n \sum_{\substack{T \subseteq [n] \\ |T|=j}} \theta^{\text{ext}}(\lambda,p)^j(1-\theta^{\text{ext}}(\lambda,p))^{n-j}\\
			&= \lambda \sum_{j=k}^n \binom{n}{j} \theta^{\text{ext}}(\lambda,p)^j (1-\theta^{\text{ext}}(\lambda,p))^{n-j} \nonumber \\
			&\hspace{5cm} \P\text{-a.s.}
			\end{align*}
			Define $$\theta'_{k,n} \equiv \theta'_{k,n}(\lambda,p) = \sum_{j=k}^n \binom{n}{j} \theta(\lambda p)^j (1-\theta(\lambda p))^{n-j}. $$
			
			From our interpretation of $C'_{k,n}$ as representing $C_{k,n}^{\text{ext}}$ of the original model, we use $\theta'_{k,n}$ instead of $\theta_{k,n}^{\text{ext}}$ in \eqref{pkndelta_expr}, and after a series of manipulations, the minimum forwarding probability obtained via this heuristic approach, $p'_{k,n,\delta}$, would be the minimum probability $p$ such that
			$$\frac{1}{\theta(\lambda)}\sum_{t=k}^{n} \sum_{j=k}^{t}\binom{n}{t}\binom{t}{j}\theta(\lambda p)^{t+j}(1-\theta(\lambda p))^{n-j} \ge 1-\delta. $$
			
			This expression is similar to the expression that was obtained for the case of a grid in \cite{ToN}. Using \cite[Prop. VI.11]{ToN}, we then have
			\begin{equation}
			p'_{k,n,\delta} = \inf \left\{\ p \  \Big| \  \frac{\P(Y \ge k )}{{\theta(\lambda)}} \ge 1-\delta \right\}
			\label{eq:heur_prob}
			\end{equation}
			where $Y\sim Bin(n,(\theta(\lambda p))^2)$. 
			
			The $p'_{k,n,\delta}$ values obtained using this expression is compared alongside the simulation results in Fig. \ref{fig:comparison}(a). The expected total number of transmissions obtained via \eqref{taukndelta_expr} is plotted in Fig. \ref{fig:comparison}(b).. The simulation setup is the same as described in Section \ref{sec:simu} for the intensity $\lambda=4.5$.
			\begin{figure} 
				\centering
				\subfloat[Minimum retransmission probability]{%
					\includegraphics[width=\linewidth]{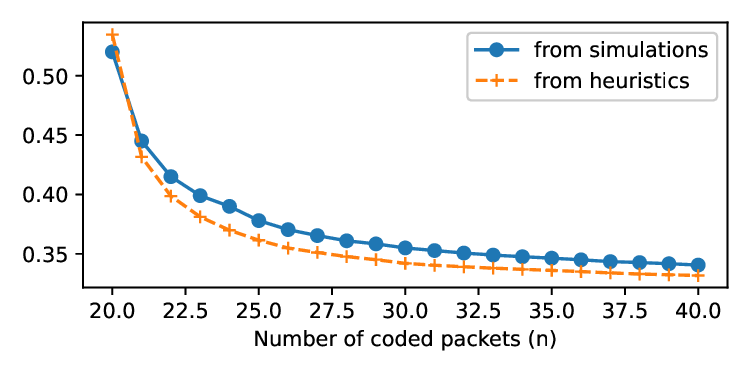}}
				\label{comp_a}
				\subfloat[Expected total number of transmissions]{%
					\includegraphics[width=\linewidth]{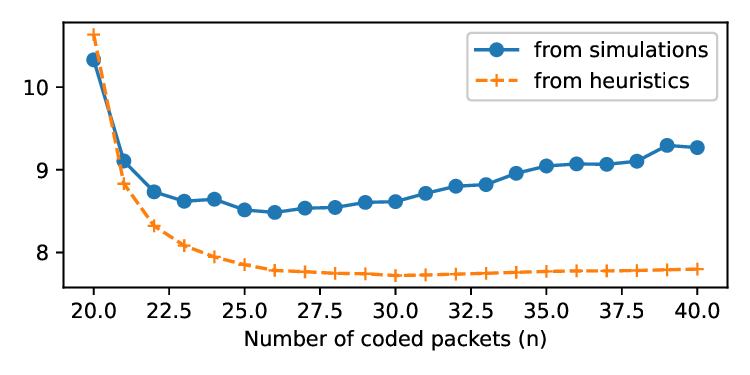}}
				\label{comp_b}\\
				\caption{Comparison of simulation results with results obtained in \eqref{eq:heur_prob} and  \eqref{taukndelta_expr} on $RGG(4.5,1)$ on $\G_{101}$ with $k=20$ packets and $\delta=0.1$}
				\label{fig:comparison} 
			\end{figure}
			
			It is observed that the curve for the minimum forwarding probability obtained via our analysis tracks the simulation curve pretty well. However, the curve for the expected total number of transmissions deviates from the simulation results substantially for larger values of $n$. This can be attributed to the drastic change in $\theta(\lambda)$ around the critical intensity $\lambda_c$. Even though there seems to be a minor difference in the forwarding probability of the original and the mean-field model, the behaviour of the percolation probability around $\lambda_c$ creates a huge divide between the two transmission plots in Fig \ref{fig:comparison}(b) . This behaviour is similar to what was obtained on the grid in \cite{ToN}. Nevertheless, note that the $\tau_{k,n,\delta}$ curve initially decreases to a minimum and then gradually increases with $n$ (albeit very slowly). This shows that probabilistic forwarding with coding is indeed beneficial on RGGs in terms of the number of transmissions required for a near-broadcast.
			
			\section{Discussion}\label{sec:disc}
			\subsection{Bounds on $\theta_{k,n}^{\text{ext}}$}
			\label{sec:bounds}
			We give two lower bounds for $\theta_{k,n}^{\text{ext}}(\lambda,p)$. 
			The probability $\theta_{k,n}^{\text{ext}}(\lambda,p) $ can be expressed in terms of the events $A_T^{\text{ext}}$ as follows.
			\begin{equation*}
			\theta_{k,n}^{\text{ext}}(\lambda,p) = \Palm\Bigg(\bigcup_{|T| \ge k} A_T^{\text{ext}}\Bigg) = \sum_{|T| \ge k} \Palm(A_T^{\text{ext}})
			\end{equation*}
			A simple lower bound for $\theta_{k,n}^{\text{ext}}(\lambda,p)$ can be obtained by taking the term corresponding to $T=[n]$ in the above summation.
			\begin{align*}
			\theta_{k,n}^{\text{ext}}(\lambda,p) \ge \Palm(A_{[n]}^{\text{ext}}) & = \Palm\left(\bigcap_{i=1}^n\{\0 \in C_{\infty,i}^{\text{ext}}\}\right)\\
			& \stackrel{(a)}{\ge} \prod_{i=1}^{n}\Palm\left(\0 \in C_{\infty,i}^{\text{ext}}\right)\\
			& = \Palm\left(\0 \in C_{\infty,i}^{\text{ext}}\right)^n
			\end{align*}
			Here, the inequality in (a) is via the FKG inequality since the events $\{\0 \in C_{\infty,i}^{\text{ext}}\} $ are increasing events. This gives
			\begin{equation}
			\theta_{k,n}^{\text{ext}}(\lambda,p)  \ge \theta(\lambda p)^n.
			\label{eq:lobnd1}
			\end{equation} 
			Note that this, along with Assumption \ref{assump:super_crit}, suffices to ensure that our analysis yields non-trivial results for all values of $k$ and $n$. 
			
			We now provide a second bound. For this, recall the iid marked point process $\Phi$ equipped with the mark structure $\bZ$. Define a new marked point process $\Phi_T$ with the underlying point process $\Phi$ and marks $Z_T = \prod_{i\in T}Z_i \prod_{j \notin T}(1-Z_j)$. The points with mark $1$ in $\Phi_T$, form a thinned version of $\Phi$ where each vertex is retained with probability $\P(Z_T=1|\Phi) = \P(Z_i=\1\{i \in T \}\ ,\ \ i\in [n]|\Phi) = p^{|T|}(1-p)^{n-|T|} $. Thus $\Phi_T$ is an iid marked point process with $Ber(p^{|T|}(1-p)^{n-|T|})$ marks.
			
			Let $C^{\text{ext}}(\Phi_T)$ denote the IEC of $\Phi_T$. Notice that
			\begin{equation*}
			\bigcup_{|T| \ge k}\{\0 \in C^{\text{ext}}(\Phi_T) \} \subseteq \{\0 \in C_{k,n}^{\text{ext}} \}. 
			\end{equation*}
			The probability of the event in the LHS above can be found as 
			\begin{align*}
			\Palm &\left(\bigcup_{|T| \ge k} \{\0 \in C^{\text{ext}}(\Phi_T) \}\right) \\
			& \hspace{1.5cm}= 1- \Palm\left(\bigcap_{|T| \ge k}\{\0 \notin C^{\text{ext}}(\Phi_T) \}\right)\\
			& \hspace{1.5cm} = 1- \prod_{j=k}^{n}\left(1-\theta^{\text{ext}}(\lambda, p^j(1-p)^{n-j})\right)^{\binom{n}{j}}
			\end{align*}
			Therefore, the probability $\theta_{k,n}^{\text{ext}}(\lambda,p)$ can be bounded as
			\begin{equation}
			\theta_{k,n}^{\text{ext}}(\lambda,p) \ge 1- \prod_{j=k}^{n}\left(1-\theta(\lambda p^j(1-p)^{n-j})\right)^{\binom{n}{j}}
			\label{eq:lobnd2}
			\end{equation}
			
			\subsection{A note on our assumptions}
			\label{sec:super}
			In this subsection, we provide some justifications for the assumptions made in our analysis. Our interest in this paper is to broadcast information on large networks. A basic requirement for this is that a large number of nodes in the network must be reachable from the origin. In the sub- critical regime, i.e. $\lambda < \lambda_c \approx 1.44$, the clusters are finite and small. To model large ad-hoc networks, we need the graph to be connected on a large area $\G_m$. This necessitates $\lambda$ to be in the super-critical regime and the component of the origin within $\G_m$ to be large. In the limit as $m \rightarrow \infty$, this requires that the origin be present in the infinite cluster of the underlying RGG, thus justifying Assumption \ref{assump:orig}.
			
			Further, notice that for a near-broadcast, we need the expected fraction of successful receivers to be close to $1$, i.e., $\Palmexp  \left[ \frac{|\cR_{k,n}(\cG^\0)\cap \G_m|}{\lambda \theta(\lambda) m^2} \right] \geq 1-\delta $ for some small $\delta >0$ (The denominator here is the expected number of nodes within $\G_m$ of the infinite cluster $C$.). If we would like this to hold for sufficiently large $m$, then the forwarding probability must be such that $\cR_{k,n}(\cG^\0)$ has infinite cardinality. This implies that $p$ must be such that there is an IEC during probabilistic forwarding on $\cG^\0$. Now, since existence of an IEC implies existence of an infinite cluster, the $p$ value must ensure presence of an infinite cluster. Thus $\lambda p >\lambda_c$. This justifies Assumption \ref{assump:super_crit}.
			
			It can also be seen from the simulation results in Fig. \ref{fig:RGG_simu} that $\tau_{k,n,\delta}$ is minimized when the forwarding probability is such that $\lambda p_{k,n,\delta} > \lambda_c$ or $p_{k,n,\delta} > 0.32$. Further, results obtained from our heuristic approach in Fig. \ref{fig:comparison}(a) and Fig. \ref{fig:comparison}(b) also suggest that the expected total number of transmissions is indeed minimized when operating in the super-critical regime.
			
			%
			%

			\bibliographystyle{IEEEtran}
			\bibliography{RGGrefs}
			
			\appendices
			
			\section{Palm probabilities}\label{sec:Palm}
			In this section, we prove three main propositions which are used in the analysis of the probabilistic forwarding protocol. Let $\cG \sim RGG(\lambda,1)$ be a random geometric graph on $\R^2$ defined on some probability space $(\Omega,\mathcal{F},\P)$. The underlying Poisson point process, $\Phi$, is of intensity $\lambda$. The intensity $\lambda$ is such that we operate in the super-critical region, i.e., $\lambda>\lambda_c$. Let $C \equiv C(\Phi)$ be the unique infinite cluster in $\cG$. Let $\Phi^\0 = \Phi \cup \{\0\}$ denote the Palm version of $\Phi$ and let $C(\Phi^\0)$ be the infinite cluster in it. Denote by $\Palm$, the Palm probability of the origin and $\Palmexp$, the expectation with respect to $\Palm$. We now show that the limiting fraction of vertices in $C$ within $\G_m$ remains the same with respect to both $\E$ and $\Palmexp$.
			
			\begin{prop}
				\begin{equation*}
				\lim_{m\rightarrow \infty} \Palmexp \left[\frac{|C \cap \G_m|}{m^2}\right] = 	\lim_{m\rightarrow \infty} \E \left[\frac{|C \cap \G_m|}{m^2}\right]
				\end{equation*}
				\label{prop:Palm_inf}
			\end{prop}
			\begin{proof}
				Let $C_1,C_2,\cdots,C_K$ be finite components in $\cG$ which intersect the ball of radius $1$ centered at the origin, i.e., $C_i\cap \Ball \neq \emptyset,\ \ \ \ \forall i \in \{1,2,\cdots,K\}$. Since vertices from distinct finite components $C_i$ and $C_j$, should be at least at a distance of $1$ from each other, the number of such components is bounded. In particular, $K$ is a random variable with $K \leq 7\  a.s.$.
				The infinite clusters in the $RGG(\Phi^\0,1)$ and $RGG(\Phi,1)$ models can be related in the following way:
				\begin{equation*}
				C(\Phi^\0) = 
				\begin{cases}
				C(\Phi) \cup C_1 \cup \cdots & \cup \  C_K \cup \{\0\}  \\ &\text{ if } C\cap \Ball \neq \emptyset \\
				C(\Phi) & \text{ if } C\cap \Ball = \emptyset
				\end{cases} 
				\end{equation*}
				Using this, we can write 
				\begin{align*}
				\frac{|C(\Phi^\0) \cap \G_m|}{m^2} &= \frac{|C(\Phi) \cap \G_m|}{m^2} \\
				&\  \hspace{0.5cm}+\ \sum_{i=1}^K \frac{|C_i \cap \G_m|}{m^2}\ \mathds{1}\{ C\cap \Ball \neq \emptyset\}
				\end{align*}
				Since $K \le 7 \ a.s.$ and $|C_i| < \infty$ for all $i=1,2,\cdots,K$, we have
				\begin{equation*}
				\sum_{i=1}^K \frac{|C_i \cap \G_m|}{m^2} \stackrel{m\rightarrow \infty}{\longrightarrow} 0 \ \ \ \P\text{-a.s.} .
				\end{equation*}
				Thus, we deduce that	
				\begin{equation}
				\lim_{m\rightarrow \infty} \frac{|C(\Phi^\0) \cap \G_m|}{m^2} = 	\lim_{m\rightarrow \infty} \frac{|C(\Phi) \cap \G_m|}{m^2} \ \ \ \ \ \ \P\text{-a.s.}
				\label{eq:palm_to_normal}
				\end{equation}
				Since the random variables involved are bounded by $1$, applying the dominated convergence theorem (DCT) gives the desired result.
			\end{proof}
			
			\begin{cor}	
				\begin{equation*}
				\lim_{m\rightarrow \infty} \Palmexp \left[\frac{|C \cap \G_m|}{m^2}\right] = \lambda \theta(\lambda)
				\end{equation*}
			\end{cor}
			\begin{proof}
				This directly follows from the previous proposition and \eqref{Eq:erg_cluster}.
			\end{proof}
			
			Next, consider the formulation of the marked point process described in Section \ref{sec:MPP_fwding}. Let $C^{\text{ext}} \equiv C^{\text{ext}}(\Phi)$ be the infinite extended cluster (IEC). We now show an analogue of the previous proposition for $C^{\text{ext}}$.
			
			\begin{prop}
				\begin{equation*}
				\lim_{m\rightarrow \infty} \Palmexp \left[\frac{|C^{\text{ext}} \cap \G_m|}{m^2}\right] = 	\lim_{m\rightarrow \infty} \E \left[\frac{|C^{\text{ext}} \cap \G_m|}{m^2}\right]
				\end{equation*}
				\label{prop:Palm_ext}
			\end{prop}
			\begin{proof}
				The proof is along the same lines as that in Proposition \ref{prop:Palm_inf}. 	Let $C_1,C_2,\cdots,C_K$ be finite components in $\cG^+$ which intersect the ball of radius $1$ centered at the origin, i.e., $C_i\cap \Ball \neq \emptyset,\ \ \ \ \forall i \in \{1,2,\cdots,K\}$. Here again $K \leq 7 \  a.s.$. Now, suppose that $C^+\cap \Ball \neq \emptyset $, then regardless of the mark of the origin, it is true that $C^{\text{ext}}(\Phi^\0) \subseteq C^{\text{ext}}(\Phi) \cup C_1^{\text{ext}} \cup \cdots \cup \ C_K^{\text{ext}}$ (with equality being true when the origin has mark $1$). If on the other hand  $C^+\cap \Ball = \emptyset $, then $C^{\text{ext}}(\Phi^\0) = C^{\text{ext}}(\Phi)$. 
				Using this, we can write 
				\begin{align*}
				\frac{|C^{\text{ext}}(\Phi^\0) \cap \G_m|}{m^2} &\leq \frac{|C^{\text{ext}}(\Phi) \cap \G_m|}{m^2} \\
				&\  \hspace{0.2cm}+\ \sum_{i=1}^K \frac{|C_i^{\text{ext}} \cap \G_m|}{m^2}\ \mathds{1}\{ C^+\cap \Ball \neq \emptyset\}.
				\end{align*}
				Note that, if $C_i$ is a finite cluster, then so is $C_i^{\text{ext}}$ and hence the summation on the RHS above tends to $0$ as $m\rightarrow \infty$. Since we trivially have that 
				\begin{equation*}
				\frac{|C^{\text{ext}}(\Phi) \cap \G_m|}{m^2} \leq \frac{|C^{\text{ext}}(\Phi^\0) \cap \G_m|}{m^2}, 
				\end{equation*}
				in the limit of large $m$, the fraction $\frac{|C^{\text{ext}}(\Phi^\0) \cap \G_m|}{m^2} $ is sandwiched between the two limits yielding
				\begin{equation*}
				\lim_{m\rightarrow \infty} \frac{|C^{\text{ext}}(\Phi^\0) \cap \G_m|}{m^2} = 	\lim_{m\rightarrow \infty} \frac{|C^{\text{ext}}(\Phi) \cap \G_m|}{m^2} \ \ \ \ \ \ \P\text{-a.s.}
				\end{equation*}
				Using DCT gives the statement of the proposition.
			\end{proof}
			
			A similar argument extends to $C_{k,n}^{\text{ext}}$ as well, which is stated in the following proposition.
			
			\begin{prop}
				\begin{equation*}
				\lim_{m\rightarrow \infty} \Palmexp \left[\frac{|C_{k,n}^{\text{ext}} \cap \G_m|}{m^2}\right] = 	\lim_{m\rightarrow \infty} \E \left[\frac{|C_{k,n}^{\text{ext}} \cap \G_m|}{m^2}\right]
				\end{equation*}
				\label{prop:Palm_ext_mult}
			\end{prop}
			\begin{proof}
				Firstly, note that 
				\begin{equation}
				\frac{|C_{k,n}^{\text{ext}}(\Phi^\0) \cap \G_m|}{m^2} \geq \frac{|C_{k,n}^{\text{ext}}(\Phi) \cap \G_m|}{m^2}.
				\label{eq:mult_ext_lb}
				\end{equation}
				The nodes in $C_{k,n}^{\text{ext}}(\Phi^\0)$ can be related to those in $C_{k,n}^{\text{ext}}(\Phi)$ in the following way. Let $C_{1}^+,C_{2}^+,\cdots,C_{n}^+$ denote the infinite clusters corresponding to each of the $n$ packets and let $C_{i,1},C_{i,2},\cdots,C_{i,K_i}$ denote the finite clusters corresponding to the $i-$th packet which intersect the ball of radius $1$ at the origin. Here again, $K_i \leq 7 \ a.s. $ for all $ i$. Proceeding with similar reasoning as that of Proposition \ref{prop:Palm_ext}, we can obtain
				
				\begin{align}
				\frac{|C _{k,n}^{\text{ext}}(\Phi^\0) \cap \G_m|}{m^2} &\leq \frac{|C _{k,n}^{\text{ext}}(\Phi) \cap \G_m|}{m^2} \nonumber\\
				&\  \hspace{0.2cm}+\ \sum_{\substack{i \in [n] \\ C_{i}^+ \cap \Ball \neq \emptyset}} \sum_{j=1}^{K_i} \frac{|C_{i,j}^{\text{ext}} \cap \G_m|}{m^2}
				\label{eq:mult_ext_ub}
				\end{align}
				The summation on the RHS is a finite sum with at most $7n$ terms with each term consisting of fraction of nodes in some finite cluster. By taking limits as $m\rightarrow \infty$, this fraction vanishes. Therefore
				the fraction $\frac{|C_{k,n}^{\text{ext}}(\Phi^\0) \cap \G_m|}{m^2} $ is sandwiched between the two limits in \eqref{eq:mult_ext_lb} and \eqref{eq:mult_ext_ub} yielding
				\begin{equation*}
				\lim_{m\rightarrow \infty} \frac{|C_{k,n}^{\text{ext}}(\Phi^\0) \cap \G_m|}{m^2} = 	\lim_{m\rightarrow \infty} \frac{|C_{k,n}^{\text{ext}}(\Phi) \cap \G_m|}{m^2} \ \ \ \ \ \ \P\text{-a.s.}
				\end{equation*}
				Using DCT gives the statement of the proposition.
			\end{proof}
			
			\section{Russo-Seymour-Welsh (RSW) theory}
			\label{app:RSW}
			Let $\Phi$ be a homogeneous Poisson point process of intensity $\lambda >\lambda_c$ on the whole $\R^2$ plane. On a box $B_{a,b}=[0,b]\times[0,a]$, a \emph{left-right crossing} of $B_{a,b}$ is defined as a sequence of vertices $\{X_i, i = 1,2,\cdots,s\}$, such that\footnote{Here $||\cdot||$ is the $L^2$ norm.} $||X_i-X_{i-1}|| \le 1$ for $i=2,3,\cdots,s$ and $||X_1-x|| \le 1$ and $||X_s-y|| \le 1$ for some $x\in \{0\} \times[0,a]$ and $y \in \{b\} \times [0,a]$. A \emph{top-bottom crossing} is defined similarly but with $x\in[0,b]\times \{0\}$ and $y\in[0,b]\times \{a\}$. If $a=b$, we simply denote the square box by $B_a$.
			
			Define $LR(a)$ to be the event that there is a left right crossing in a rectangular box $R_a=B_{a,2a}=[0,2a]\times[0,a]$. The  probability of $LR(a)$ in the super-critical region is exponentially close to $1$ as formalized in \cite[Lemma 10.5]{penrose2003random}. We reproduce the same here.
			\begin{lem}
				For $\lambda > \lambda_c$, there exists $c>0$ and $a_1>0$ such that $1-\P(\text{LR}(a)) \le \exp(-ca)$ for all $a\ge a_1$.
				\label{lem:left_right}
			\end{lem}
		We will use this lemma to obtain the probability of a left-right crossing in a $s_m \times m$ rectangular box, where $s_m \ll m$. Let $CR$ be the event that there is a left-right crossing of the box $B=[0,m]\times [0,s_m]$. We then have the following proposition.
		
		\begin{prop}
			For $\lambda>\lambda_c$, there exists $c>0$ and $a_1>0$ such that $1-\P(CR) \le 2 \left\lceil \frac{m}{s_m}\right\rceil \exp(-cs_m)$ for all $s_m\ge a_1$.
			\label{lem:crossing}
		\end{prop}
		\begin{proof}
			Denote $\ell = \left\lceil \frac{m}{s_m}\right\rceil$. Let $B_i = [(i-1)s_m,is_m]\times [0,s_m]$ for $i\in \{1,2,\cdots,\ell\}$ and let $R_j = B_j \cup B_{j+1}$ for $j\in \{1,2,\cdots,\ell-1\}$ (see Fig. \ref{fig:cross}). Define $LR_j$ to be the event that there is a left-right crossing in $R_j$ and let $TB_i$ be the event that there is a top-bottom crossing of $B_i$. Notice that
			\begin{figure}
				\centering
				\includegraphics[width=0.5\textwidth]{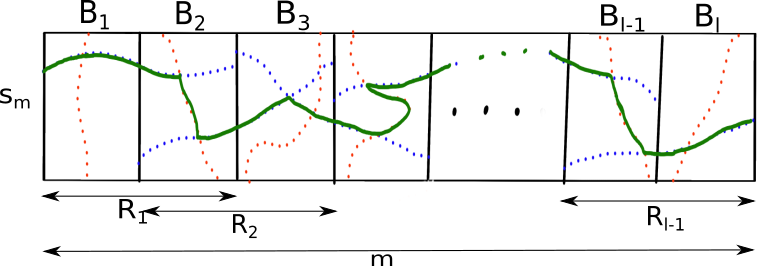}
				\caption{Left-right crossing in the $s_m \times m$ rectangular box $B$ through the events $LR_j$ and $TB_i$.}
				\label{fig:cross}
			\end{figure}%
			$$ CR \supseteq \bigcap_{i=1}^\ell TB_i \cap \bigcap_{j=1}^{\ell-1} LR_j,$$
			which gives 
			\begin{equation*}
			\P(CR^\mathsf{c}) 
			\le \sum_{i=1}^\ell \P(TB_i^\mathsf{c})\  +\  \sum_{j=1}^{\ell-1} \P(LR_j^\mathsf{c}).
			\end{equation*}
		The probability of there being no left-right crossings in the rectangles $R_j$, for $j \in \{1,2,\cdots,\ell-1 \}$, are identical (due to translation invariance) and hence the latter term in the above expression can be replaced by $(\ell-1)\P(LR_1^\mathsf{c})$. For the first term, note that absence of a top-bottom crossing of $B_i$ implies that there is no top-bottom crossing in the rectangle $R_i'=[(i-1)s_m,is_m]\times [0,2s_m]$. But a top-bottom crossing in $R_i'$ is the same as a left-right crossing in $R_1$ (say), since the underlying homogeneous Poisson point process $\Phi$ is isotropic. This gives 
			\begin{equation*}
			\P(CR^\mathsf{c}) \le (2\ell-1)P(LR_1^\mathsf{c}),
			\end{equation*}
			which from Lemma \ref{lem:left_right} gives the statement of the proposition.
		\end{proof}
		\begin{figure}
			\centering
			\includegraphics[width=0.4\textwidth]{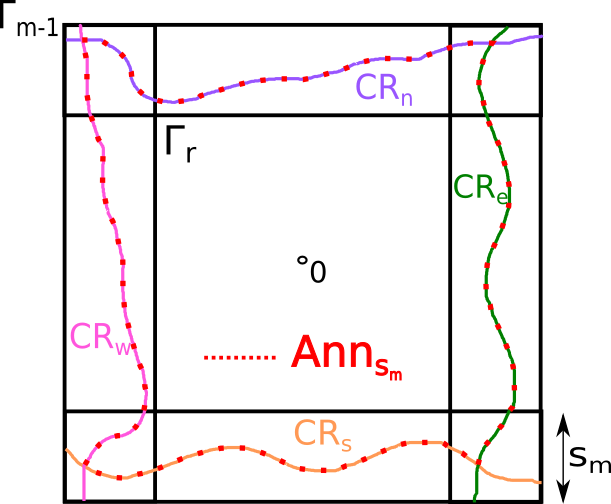}
			\caption{Circuit formed by the four left-right crossings $LR_d, d\in \{n,s,e,w \}$}
			\label{fig:annulus}
		\end{figure}%
			Next, we apply Proposition \ref{lem:crossing} to the four rectangles surrounding $\G_r$ as depicted in Fig. \ref{fig:annulus} . Let $CR_d$ for $d \in \{n,s,e,w \}$ be the event denoting the existence of crossings inside the four rectangles and let $\text{Ann}_{s_m}$ be the event that there is a circuit in the annulus $\G_{m-1} \setminus \G_r$ as shown in Fig. \ref{fig:annulus}. Since the presence of crossings in the four rectangles ensures the occurence of $\text{Ann}_{s_m}$, we obtain
			\begin{align*}
			\P(\text{Ann}_{s_m}^\mathsf{c}) & \le \P\left(\bigcup_d CR_d^\mathsf{c}\right),\\
			&\le \sum_d\P(CR_d^\mathsf{c}),\\
			&\le 8\left\lceil \frac{m}{s_m}\right\rceil \exp(-cs_m).
			\end{align*}
			We state this formally in the following proposition.
			
			\begin{prop}
				For $\lambda > \lambda_c$, there exists $c>0$ and $a_1>0$ such that $1-\P(\text{Ann}_{s_m}) \le 8\left\lceil \frac{m}{s_m}\right\rceil \exp(-cs_m)$ for all $s_m\ge a_1$.
				\label{prop:annulus}
			\end{prop}	
			
			\textbf{Remark:} Note that the statement of the above propostition holds even with respect to the Palm probability $\Palm$. This is because introducing a point at the origin does not affect the event $\text{Ann}_{s_m}$, and hence $\Palm(\text{Ann}_{s_m})=\P(\text{Ann}_{s_m})$.

			\section{Proof of Theorem \ref{thm:recs}}
			\label{app:proof_recs}
			\begin{thm}[Restatement of Theorem \ref{thm:recs}]
				For $\lambda p >\lambda_c$, we have \
				\begin{equation*}
				\lim_{m\rightarrow \infty} \MarkPalm_A  \left[ \frac{R_{k,n}(G_m)}{|C_\0(G_m)|} \right] = 
				\ \frac{1}{\theta(\lambda)^2}\sum_{t=k}^n \binom{n}{t} \theta_{k,t}^{\text{ext}}\  \ProbMark (A_{[t]}^{\text{ext}}).
				\end{equation*}
			\end{thm}
			\begin{proof}[Sketch of proof:]
				\textbf{Step 1:} We first evaluate $$\lim_{m\rightarrow \infty}\MarkPalm  \left[ \frac{R_{k,n}(G_m)}{|C_\0(G_m)|} \1_A \right] = \lim_{m\rightarrow \infty}\MarkPalm  \left[ \frac{R_{k,n}(G_m) \1_A}{|C_\0(G_m)| \1_A}  \right]$$ and then divide it by $\P(A) = \P(\0 \in C(\cG^0)) = \theta(\lambda)$ to obtain the required conditional expectation.  We take the convention that $\frac{0}{0}=0$. Note that Assumption \ref{assump:super_crit} ensures that $\theta(\lambda)>0$.
				
				\textbf{Step 2:} Specializing the statement of Theorem \ref{thm:rgg_out_recs} on the event $A$, we obtain
					\begin{equation*}
					\lim\limits_{m\rightarrow \infty} \frac{|C_\0(G^\0_m)|}{\lambda m^2} \1_A = \lim\limits_{m\rightarrow \infty} \frac{|C_\0(\cG^\0) \cap \G_m |}{\lambda m^2} \1_A \hspace{1cm}  \P\text{-a.s.}.
					\end{equation*}
					Notice that on the event $A$, $C_\0(\cG^\0) =C(\cG^\0)$. Using  \eqref{eq:palm_to_normal}, \eqref{Eq:erg_cluster} and the note following Lemma \ref{lem:comp_org}, we have for $\lambda > \lambda_c$
					\begin{align*}
					\lim\limits_{m\rightarrow \infty} \frac{|C_\0(G_m)|}{\lambda m^2} \1_A &= \lim\limits_{m\rightarrow \infty} \frac{|C(\cG) \cap \G_m |}{\lambda m^2} \1_A \\
					&= \theta(\lambda) \1_A \hspace{1.5cm}  \P^\0\text{-a.s.}.
					\end{align*}
					Conditional on the mark of the origin $\bZ(\0)=\1$, we have 
					$$\lim\limits_{m\rightarrow \infty} \frac{|C_\0(G_m)|}{\lambda m^2} \1_A  = \theta(\lambda) \1_A  \hspace{1.5cm}  \ProbMark \text{-a.s.} $$

				\textbf{Step 3:} Let $\cR_{k,n}(\cG) $ be the set of nodes that receive at least $k$ out of the $n$ packets from the origin when probabilistic forwarding is carried out on $\cG$. Using arguments similar to those of Theorem \ref{thm:rgg_out_recs} for nodes without $\G_m$-conduits, we have that
				\begin{align}
				\lim_{m\rightarrow \infty} \MarkPalm &\left[ \frac{R_{k,n}(G_m)}{\lambda m^2} \1_A \right] = \nonumber\\
				& \hspace{1cm} \lim_{m\rightarrow \infty} \MarkPalm  \left[ \frac{|\cR_{k,n}(\cG)\cap \G_m|}{\lambda m^2} \1_A \right].
				\label{eq:step3}
				\end{align}
				
				\textbf{Step 4:} For $T \subseteq [n]$, let $A_T^{\text{ext}}$ be the event that the origin is present in exactly the IECs indexed by $T$. Conditioning on the event $A_T^{\text{ext}}$, we obtain
				\begin{align}
				&\MarkPalm  \left[ \frac{|\cR_{k,n}(\cG)\cap \G_m|}{\lambda m^2} \1_A \right] = \nonumber \\
				&\hspace{0.5cm} \sum_{t=0}^n \sum_{\substack{T\subseteq [n] \\ |T|=t}} \MarkPalm  \left[ \frac{|\cR_{k,n}(\cG)\cap \G_m|}{\lambda m^2} \1_A \Bigg| A_T^{\text{ext}} \right] \ProbMark (A_T^{\text{ext}}).
				\label{eq:conditioning}
				\end{align}
				If $|T| < k$, then the nodes of $\cR_{k,n}(\cG) $ within $\G_m$ must reside in finite clusters whose fraction vanishes in the limit of large $m$. If $|T|\geq k$, then it is only the nodes which are within at least $k$ IECs among those packet transmissions which are indexed by $T$, that contribute towards the expectation. Denote such nodes by $\cR_{k,T}$. The remaining nodes of $\cR_{k,n}(\cG) $ within $\G_m$, must be in at least one finite cluster and hence their fraction vanishes in the limit. Additionally, given $A_T^{\text{ext}}$ for $|T|>0$, the $\0$ must be present in the infinite cluster of the underlying graph i.e., $\1_A=1$. Putting all these together, we obtain
				\begin{align}
				\lim_{m\rightarrow \infty} & \MarkPalm  \left[ \frac{|\cR_{k,n}(\cG)\cap \G_m|}{\lambda m^2} \1_A\right] = \nonumber \\
				&\lim_{m\rightarrow \infty} \sum_{t=k}^n\sum_{\substack{T\subseteq [n] \\ |T|= t}} \MarkPalm  \left[ \frac{|\cR_{k,T}\cap \G_m|}{\lambda m^2} \Bigg{|} A_T^{\text{ext}} \right] \ProbMark (A_T^{\text{ext}}).
				\label{eq:sum_expr}
				\end{align}
				
				\textbf{Step 5:} Define $\mathbf{O}$ to be the event that the origin has mark $1$ in all the $n$ packet transmissions. The expectation on the RHS in the above equation can be written as 
				\begin{align*}
				\MarkPalm  \left[ \frac{|\cR_{k,T}\cap \G_m|}{\lambda m^2} \Bigg{|} A_T^{\text{ext}} \right]& =\\
				& \Palmexp \left[ \frac{|\cR_{k,T}\cap \G_m|}{\lambda m^2} \Bigg{|} A_T^{\text{ext}} \cap \mathbf{O} \right].
				\end{align*}
				$\cR_{k,T}$ is independent of the packet transmissions which are not in $T$. The event $\mathbf{O}$ can be thus restricted to only those indices in $T$. However, the conditioning event $A_T^{\text{ext}} \cap \mathbf{O}$ is then the event that $\0$ is in the infinite cluster $C^+$ in the packet transmissions indexed by $T$. Call this event $A_T^+$. We then have
				\begin{equation}
				\MarkPalm  \left[ \frac{|\cR_{k,T}\cap \G_m|}{\lambda m^2} \Bigg{|} A_T^{\text{ext}} \right] = \Palmexp \left[ \frac{|\cR_{k,T}\cap \G_m|}{\lambda m^2} \Bigg{|} A_T^+ \right]
				\label{eq:exttoplus}
				\end{equation} 
				
				\textbf{Step 6:}
				Conditional on the event $A_T^+$, the set $\cR_{k,T}$ has the same distribution as the set $C_{k,|T|}^{\text{ext}}$,  which was defined in Section \ref{sec:MPP_mult}. This gives
				\begin{equation*}
				\Palmexp \left[ \frac{|\cR_{k,T}\cap \G_m|}{\lambda m^2} \Bigg{|} A_T^+ \right] =  \Palmexp \left[ \frac{| C_{k,|T|}^{\text{ext}}\cap \G_m|}{\lambda m^2}\right].
				\end{equation*}
				From Proposition \ref{prop:Palm_ext_mult}, by taking limits as $m\rightarrow \infty$, the expectation with respect to the Palm probability, $\E^\0$, can be written in terms of the expectation $\E$, yielding
				\begin{equation}
				\lim_{m \rightarrow \infty} \Palmexp \left[ \frac{|\cR_{k,T}\cap \G_m|}{\lambda m^2}  \Bigg{|} A_T^+ \right] = \lim_{m \rightarrow \infty} \E \left[ \frac{| C_{k,|T|}^{\text{ext}}\cap \G_m|}{\lambda m^2}  \right]
				\label{eq:PalmtoE}
				\end{equation}
				
				\textbf{Step 7:} Using \eqref{succ_recs} with $n$ replaced by $|T|=t$ and employing DCT, we obtain 
				\begin{equation}
				\lim_{m \rightarrow \infty} \E \left[ \frac{| C_{k,|T|}^{\text{ext}}\cap \G_m|}{\lambda m^2} \right] = \theta_{k,t}^{\text{ext}}(\lambda,p) 
				\label{eq:theta}
				\end{equation}
				
				\textbf{Step 8:} Clubbing the expressions from \eqref{eq:exttoplus}, \eqref{eq:PalmtoE} and \eqref{eq:theta} into \eqref{eq:sum_expr}, and using \eqref{eq:step3}, we obtain
				\begin{equation*}
				\lim_{m\rightarrow \infty} \MarkPalm \left[ \frac{R_{k,n}(G_m)}{\lambda m^2}\1_A \right] = \sum_{t=k}^n \sum_{\substack{T\subseteq [n] \\ |T|=t}} \ \theta_{k,t}^{\text{ext}}\  \ProbMark (A_T^{\text{ext}}).
				\end{equation*}
				
				\textbf{Step 9:} The event $A_T^{\text{ext}}$ can be expressed as $$A_T^{\text{ext}} = \bigcap_{i \in T} \{\0 \in C_i^{\text{ext}}\} \bigcap_{j \notin T} \{\0 \notin C_j^{\text{ext}} \}. $$ Here,  $C_1^{\text{ext}}, C_2^{\text{ext}},\cdots, C_n^{\text{ext}}$ denote the IECs corresponding to the $n$ packet transmissions. Since $\{ \0 \in C_i^{\text{ext}} \} = \{\Ball \cap C_i^+ \neq \emptyset \}$, the event $A_T^{\text{ext}}$ does not depend on the specific mark of $\0$. Furthermore, the event $A_T^{\text{ext}}$ does not depend on the specific choice of the set $T$, but just on the cardinality $|T|$. This is because a relabeling of the packets does not alter the probability of $A_T^{\text{ext}}$. For a particular value of $|T| = t$, define $$A_{[t]}^{\text{ext}} = \bigcap_{i=1}^t \{\0 \in C_i^{\text{ext}}\} \bigcap_{j=t+1}^n \{\0 \notin C_j^{\text{ext}} \}.$$  Notice now that the terms within the summation in Step 7, $\theta_{k,t}^{\text{ext}}\  \ProbMark (A_{T}^{\text{ext}})$ are identical for different $T$ with the same cardinality. Therefore,
				\begin{equation*}
				\lim_{m\rightarrow \infty} \MarkPalm \left[ \frac{R_{k,n}(G_m)}{\lambda m^2}\1_A \right] = \sum_{t=k}^n \binom{n}{t} \ \theta_{k,t}^{\text{ext}}\  \ProbMark (A_T^{\text{ext}}).
				\end{equation*}
				
				\textbf{Step 10:} Putting together the results from Step 2 and Step 9 and dividing by $\theta(\lambda)$ gives the statement of the theorem.
			\end{proof}

		\end{document}